\documentclass[sn-mathphys,Numbered]{sn-jnl}% Math and Physical Sciences Reference Style
\usepackage{geometry}
\usepackage{leftidx}
\usepackage{graphicx}%
\usepackage{multirow}%
\usepackage[integrals]{wasysym}
\usepackage{amsmath,amssymb,amsfonts}%
\usepackage{amsthm}%
\newtheorem{theorem}{Theorem}%  meant for continuous numbers
\newtheorem{proposition}{Proposition}% 
\newtheorem{remark}{Remark}%
\newtheorem{lemma}{Lemma}
\newtheorem{conjecture}{Conjecture}
  \newtheorem{definition}{Definition}%
\usepackage{mathtools}
\usepackage{mathrsfs}%
\usepackage[title]{appendix}%
\usepackage{xcolor}%
\usepackage{textcomp}%
\usepackage{manyfoot}%
\usepackage{booktabs}%
\usepackage{algorithm}%
\usepackage{algorithmicx}%
\usepackage{algpseudocode}%
\usepackage{listings}%
\usepackage{enumerate}
\usepackage{mathptmx}
\usepackage{mathrsfs}
\DeclareMathAlphabet{\mathcal}{OMS}{cmsy}{m}{n}
\usepackage{orcidlink}
\usepackage{hyperref}
\usepackage{ulem}
\hypersetup
    {
    colorlinks=true,
    linkcolor=blue,
    filecolor=red,      
    urlcolor=blue,
    linktoc=page,
    citecolor=magenta,
    }
    \usepackage{manyfoot}

\begin{document}

\title[Article Title]{Tipler naked singularities in $N$ dimensions}
%\title{Tipler Naked Singularities in D-Dimensions }
%


\author{Kharanshu N. Solanki$^{\flat}$ \orcidlink{0000-0002-6909-4370}, Karim Mosani$^{\dagger}$ \orcidlink{0000-0001-5682-1033}, Omkar Deshpande$^{\sharp}$ \orcidlink{0009-0006-5745-8081} and Pankaj S. Joshi$^{\star}$ \orcidlink{0000-0003-0463-7554}} 

\affil{$^\flat$ Department of Physics, Universität Wien, Boltzmanngasse 5, 1090 Wien %\footnotetext{$^{\flat}$ Corresponding author e-mail: \url{kharanshu.solanki1889@gmail.com}} 
}
\affil{$^{\dagger}$ Fachbereich Mathematik, Universit\"{a}t T\"{u}bingen, Auf der Morgenstelle 10, 72076 T\"{u}bingen \footnotetext{$^{\dagger}$ Now at Fakultät für Mathematik, Ruhr-Universität Bochum, Universitätsstr. 150, 44780 Bochum.}} 
\affil{$^{\sharp}$ Mathematisches Institut, Ludwig-Maximilians-Universit\"{a}t, Theresienstr. 39, 80333 München} 
\affil{$^{\flat,\star}$ International Center for Space and Cosmology, Ahmedabad University, 380009 Ahmedabad}

\email{kharanshu.solanki1889@gmail.com, kmosani2014@gmail.com, omkar.deshpande@campus.lmu.de, pankaj.joshi@ahduni.edu.in}
\date{July 15, 2024}

\abstract{A spacetime singularity, identified by the existence of incomplete causal geodesics in the spacetime, is called a (Tipler) strong curvature singularity if the volume form acting on independent Jacobi fields along causal geodesics vanishes in the approach of the singularity. It is called naked if at least one of these causal geodesics is past incomplete. Here, we study the formation of strong curvature naked singularities arising from spherically symmetric gravitational collapse of general type-I matter fields in an arbitrarily finite number of dimensions. In the spirit of Joshi and Dwivedi (1993 \textit{Phys. Rev. D} \textbf{47} 5357), and Goswami and Joshi (2007 \textit{Phys. Rev. D} \textbf{76} 084026), beginning with regular initial data, we derive two distinct (but not mutually exclusive) conditions, which we call the \textit{positive root condition} (PRC) and the \textit{simple positive root condition} (SPRC), that serve as necessary and sufficient conditions, respectively, for the existence of naked singularities. In doing so, we generalize the results of both the aforementioned works. We further restrict the PRC and the SPRC by imposing the curvature growth condition (CGC) of Clarke and Krolak (1985 \textit{J. Geom. Phys.} \textbf{2}(2) 127) on all causal curves that satisfy the causal convergence condition. \textcolor{black}{The CGC then gives a sufficient condition ensuring that the naked singularities implying the PRC and implied by the SPRC are of strong curvature type. Further, due to Ricci curvature blow up, these also correspond to $C^2$ inextendibility.} Using the CGC, we extend the results of Mosani et. al. (2020 \textit{Phys. Rev. D} \textbf{101} 044052) (that hold for dimension $N=4$) to the case $N=5$, showing that strong curvature naked singularities can occur in this case. However, for the case $N\geq6$, we show that past-incomplete causal curves that identify naked singularities do not satisfy the CGC. These results shed light on the validity of the cosmic censorship conjectures in arbitrary dimensions.
}
%A spacetime singularity, identified by the existence of incomplete causal geodesics, is termed a (Tipler) strong curvature singularity if the volume form vanishes along independent Jacobi fields near the singularity. It's considered naked if it marks the past endpoint of causal curves. This study examines the emergence of strong curvature naked singularities from spherically symmetric gravitational collapse in a finite number of dimensions. Following the approach of Joshi and Dwivedi, and Goswami and Joshi, regular initial data are constructed based on matter and geometric variables satisfying energy conditions. Two conditions, the positive root condition (PRC) and the simple positive root condition (SPRC), are established as necessary and sufficient for naked singularity existence, expanding on previous research. These conditions are further refined by imposing the curvature growth condition (CGC) of Clarke and Krolak on causal curves satisfying the causal convergence condition. The CGC ensures the identified naked singularities are of strong curvature type, indicating spacetime's inextendibility. Extending prior findings to five dimensions, strong curvature naked singularities are demonstrated to occur. However, for dimensions six and above, past-incomplete causal curves identifying naked singularities fail to meet the CGC. These outcomes contribute insights into the validity of cosmic censorship conjectures across dimensions.


%%\pacs[JEL Classificat51

 
\maketitle
\tableofcontents

\section{Introduction}\label{sec1}

\noindent  The modern formulations of the cosmic censorship conjectures are built on the dynamical formulation of general relativity, indicating that these conjectures (in a broad sense) speak about the nature of singularities that form dynamically in general relativity. Fundamental to the idea of dynamics in general relativity is the notion of \textit{initial data sets}. These are $5$-tuples $(\Sigma,\leftidx{^\Sigma}{\mathbf{g}}{},\leftidx{^\Sigma}{\mathbf{K}}{},\mu, J)$, where $\Sigma$ is the initial spacelike hypersurface with first and second fundamental form fields $\leftidx{^\Sigma}{\mathbf{g}}{}$ and $\leftidx{^\Sigma}{\mathbf{K}}{}$; and $\mu \in \mathcal{C}^k (\Sigma)$ and $J\in \Gamma (\mathcal{T}\Sigma)$ are the energy density and momentum density of the matter field considered respectively, that are to be defined on $\Sigma$ with an appropriate regularity class. This initial data set must satisfy a certain set of constraint equations \cite{Ringstrom:2015jza}, subject to which it is to be evolved. \\

\noindent Employing this formalism, Choquet-Bruhat \cite{1952AcMa...88..141F} developed the fundamental theorem of dynamics in general relativity in 1952. The theorem says that \textit{given a vacuum initial data set $(\Sigma, \leftidx{^\Sigma}{\mathbf{g}}{}, \leftidx{^\Sigma}{\mathbf{K}}{}, \mu=0, J=0)$ solving the vacuum constraint equations, one can always find a spacetime $(M,\mathbf{g})$ that (i) solves the vacuum Einstein equations and (ii) admits $\Sigma \hookrightarrow M$ as a Cauchy hypersurface}. In other words, general relativity can be posed as a well-defined Cauchy problem. Later, in 1969, Choquet-Bruhat and Geroch \cite{Choquet-Bruhat:1969ywq} proved that\footnote{\noindent This proof appeals to Zorn's lemma, which is equivalent to the set-theoretic axiom of choice. This may seem to indicate that the dynamics of general relativity rely on the axiom of choice. However, in 2016, Sbierski \cite{Sbierski:2013kca} carried out an alternative proof of the uniqueness of globally hyperbolic solutions without employing Zorn's lemma.} \textit{a vacuum initial data set $(\Sigma, \leftidx{^\Sigma}{\mathbf{g}}{}, \leftidx{^\Sigma}{\mathbf{K}}{}, \mu=0, J=0)$ solving the vacuum constraint equations admits a unique vacuum maximal globally hyperbolic development $M$.}\\

\noindent Before the advent of the dynamical formulation of relativity, the general consensus about singular solutions of the Einstein equations was to discard them as ``bad" solutions. However, in the context of the dynamical formulation of relativity, one is not allowed to discard a solution, but can only make judgements about the initial data (which is regular in the first place) \cite{Dafermos:2008en}. For instance, one may argue that perhaps the singular behaviour of the Schwarzschild geometry is only an artefact of the spherically symmetric initial data that it evolves from. However, Hawking and Penrose \cite{Hawking:1970zqf} showed that such arguments must be considered invalid in the context of their singularity or incompleteness theorems. These theorems show the geodesic incompleteness of spacetime under the satisfaction of sufficiently general constraints such as certain mild causality conditions, energy conditions and appropriate initial or boundary conditions \cite{Senovilla:1998oua,Steinbauer_2022}.\\

\noindent At the root of the Hawking-Penrose singularity theorems of general relativity is the notion of causal geodesic incompleteness, and not that of unbounded curvatures along the geodesic or even the inextendibility of the spacetime in question. This leads to a rather improper terminology for the notion of singularities, which in the context of the singularity theorems, indicate the existence of incomplete geodesics in the spacetime (even though the spacetime itself may admit an extension of suitable regularity class). %\textcolor{black}{While Hawking \cite{hawking2023large} had already clarified this in the 1970s, Kerr \cite{Kerr:2023rpn} also recently emphasized this by constructing explicit examples for the Kerr geometry.} 
Additionally, the singularity theorems make no claim whatsoever about the existence of black holes. In other words, the theorems do not negate the existence of naked singularities \cite{1974IAUS...63..263P}. However, concerns regarding the nature of naked singularities and their relation with the breakdown of classical predictability, led Penrose \cite{Penrose:1969pc} to introduce the \textit{cosmic censorship conjecture} (CCC). In present literature \cite{Wald:1984rg,chrusciel2020geometry,Dafermos:2008en,Kommemi:2011wh}, this conjecture is often expressed as two separate statements. In the context of dynamical relativity, these conjectures take the following forms:\\
\textcolor{black}{\begin{conjecture}
(Weak cosmic censorship - Choquet-Bruhat \cite{choquet2009general})\\
The maximal globally hyperbolic development (MGHD) $(M,\mathbf{g})$ of generic asymptotically flat initial data  $(\Sigma,\leftidx{^\Sigma}{\mathbf{g}}{},\leftidx{^\Sigma}{\mathbf{K}}{},\mu, J)$ possesses a complete\footnote{\textcolor{black}{Christodolou \cite{Demetrios} defined that for an asymptotically flat MGHD to possess a complete null infinity $\mathscr{I}^+$, the null geodesic generators of $\mathscr{I}^+$ must be complete with respect to the conformal metric of the maximal development.}} future null infinity.\\
\end{conjecture}
\begin{conjecture}
    (Strong cosmic censorship - Choquet-Bruhat \cite{choquet2009general})\\
    The MGHD of generic asymptotically flat initial data  $(\Sigma,\leftidx{^\Sigma}{\mathbf{g}}{},\leftidx{^\Sigma}{\mathbf{K}}{},\mu, J)$ is inextendible\footnote{\textcolor{black}{Wald \cite{Wald:1984rg} adds to this saying that if the MGHD is extendible, then for each point $p$ on the Cauchy horizon $H^{+}(\Sigma)$ in the extension, either strong causality is violated, or $\overline{I^{-}(p)\cap \Sigma}$ is non-compact.}} as a suitably regular Lorentzian manifold.\\
\end{conjecture}
}

\noindent The terminology ``weak" and ``strong" is rather unfortunate, since strong cosmic censorship (SCC) does not imply weak cosmic censorship (WCC). That being said, WCC is (in some sense) weaker than SCC because the former allows the singularity to be \textit{visible to nearby observers}, while the latter does not allow the singularity to be \textit{visible to any observer}. Singularities that violate WCC are termed \textit{globally naked}, while those that violate SCC are termed \textit{locally naked}. Fig. \ref{fig1} illustrates these notions for the Lemaitre-Tolman-Bondi (LTB) inhomogeneous dust collapse scenario \cite{Joshi:2013xoa}.\\ 

\noindent The current form of the CCCs indicate that they have been introduced to resolve two issues. The first issue concerns the existence of naked singularities and their relation to the breakdown of predictability \cite{PhysRevD.14.2460}; while the second one concerns the confusion regarding the meaning of the word \textit{singularity} as demanded by the Hawking-Penrose singularity theorems. Indeed, in their current form, the cosmic censorship conjectures deal with the existence of maximal globally hyperbolic inextendible developments of generic initial data. In other words, incomplete geodesics must be indicative of genuine spacetime singularities through which the spacetime itself is inextendible. Additionally, these singularities must be covered by a horizon.\\

\begin{figure}[t]
 			\begin{center}
				\includegraphics[width=120mm,scale=0.5]{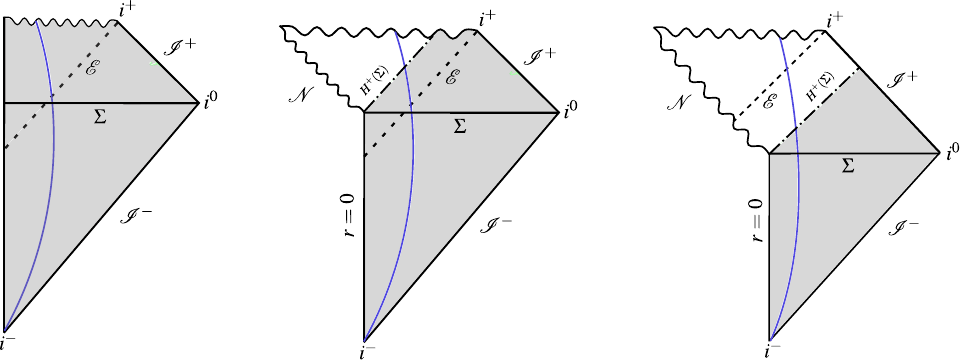}
			\end{center}
   \vspace{1em}
       		\caption{ \textcolor{black}{On the LHS, we have the conformal diagram for the well known Oppenheimer-Synder-Datt homogeneous dust collapse, which results in the formation of a black hole. Moreover, in the center we have the conformal diagram for the LTB inhomogeneous dust collapse scenario with a locally naked singularity, whereas on the RHS we have the conformal diagram for the LTB inhomogeneous dust collapse scenario with a globally naked singularity \cite{PhysRevD.19.2239,PhysRevD.109.064019}. The shaded regions represent the MGHDs of some initial data hypersurface $\Sigma$. Further, $H^{+}(\Sigma)$ denotes the future Cauchy horizon of $\Sigma$, $\mathscr{E}$ denotes the event horizon, and $\mathscr{N}$ denotes the naked singularity. Note that for the locally naked scenario, $\mathscr{I}^{+}$ is complete and $H^{+}(\Sigma) \cap \mathscr{I}^{+}=\emptyset$, whereas for the globally naked scenario, $\mathscr{I}^{+}$ is incomplete and $H^{+}(\Sigma) \cap \mathscr{I}^{+} \neq \emptyset$. Additionally for both the scenarios, the MGHDs can be extended (to the unshaded region) which then does not admit $\Sigma$ as a Cauchy surface, and hence is not globally hyperbolic. Hence, globally naked singularities violate both WCC and SCC, whereas locally naked singularities only violate the SCC.} }\label{fig1}
\end{figure}

\noindent The aforementioned precise formulations of cosmic censorship conjectures are built on the assumption that the underlying spacetime is asymptotically flat or that it is the MGHD of asymptotically flat initial data. We hold that such an assumption is rather strong. An alternate characterization of the cosmic censorship conjectures can be developed by incorporating the following requirements on physical grounds \cite{joshi2007gravitational}: (i) A suitable energy condition must be obeyed, (ii) the collapse must develop from generic regular initial data, (iii) any singularity arising from a realistic collapse must be gravitationally strong, and (iv) the matter fields must be sufficiently general.\\

\noindent A limitation of the cosmic censorship conjectures is the lack of a unique definition of the term \textit{generic initial data}. It is therefore imperative to consider specific solutions to the Einstein equations that admit naked singularities. These solutions must be studied in order to construct a relevant notion of generic initial data for the CCCs, and then make appropriate judgements about them. From the physical point of view, one may consider the formation of spacetime singularities via the unhindered gravitational collapse of a massive star that has used up its nuclear energy \cite{joshi2007gravitational}. Such studies have, by and large, been conducted by considering various forms of the matter that constitute the collapsing star. Oppenheimer and Snyder \cite{PhysRev.56.455} studied the unhindered collapse of homogeneous dust in 1939, wherein they showed that such a collapse inevitably leads to the formation of black holes, i.e., singularities that are covered by an apparent horizon. In 1984, Christodoulou \cite{Christodoulou:1984mz} showed the existence of naked singularities in inhomogeneous dust collapse. However, Newman \cite{Newman:1985gt} showed that these singularities do not satisfy the strong limiting focussing condition (SLFC) \cite{1985JGP.....2..127C}, which is a sufficient condition for the singularity to be strong (in Tipler's sense) for spherically symmetric spacetimes \cite{Tipler:1977zza}. The strength of a singularity characterizes the continuous inextendibility of the spacetime through it, and hence makes sure that the existence of incomplete geodesics is equivalent to the existence of a genuine spacetime singularity. In 1993, Joshi and Dwivedi \cite{PhysRevD.47.5357} showed the existence of naked singularities in the same setting as Christodoulou, which satisfies Tipler's criterion of strength. In 1994 \cite{Dwivedi:1994qs}, it was shown that the same results hold for all type-I fields in 4 dimensions. Recently, Mosani et al. \cite{PhysRevD.101.044052} extended these results to singularities arising in perfect fluid collapse. \\

\noindent Going beyond general relativity, one may identify the following motivations for studying gravitational collapse in theories assuming a higher dimensional spacetime: (i) It has been suggested \cite{Randall_1999} that we may be living in a lower dimensional surface (called a \textit{brane}) of a higher dimensional space. All the non-gravitational forces are confined to this brane, but gravity is not. This requires a higher dimensional generalization of Einstein gravity \cite{horowitz2005higher}. (ii) Having knowledge of certain properties in higher dimensions could help better understand those properties in four dimensions, for which general relativity was originally constructed. Goswami and Joshi developed the formalism to study the gravitational collapse of spherically symmetric type-I matter field in higher dimensions \cite{PhysRevD.76.084026} and showed the existence of naked singularities. Moreover, Giambo and Quintavalle \cite{Giambo:2007ps} have investigated how the dimension of a spacetime affects the existence of a ``phase transition"\footnote{Here, by the existence of a ``phase transition", we mean the existence of a ``critical" case where the end state of the collapse is determined by the behaviour of a single parameter. between black holes and naked singularities.} However, the strength of such singularities (in the sense of Tipler) has not been examined \cite{Nolan_2003, Nolan:1999tv, Nolan:2000rn, Nolan:1999tw, clarke1993analysis}. (iii) It may be interesting to assess the validity of the cosmic censorship conjectures in higher dimensions. In fact, it has been shown that in some special scenarios, even though cosmic censorship may be violated in four dimensional spacetimes, it is restored beyond five dimensions \cite{joshi2007gravitational}.\\

\noindent In the context of the discussion thus far, our aim here is to probe the conditions for the occurrence of strong curvature naked singularities from spherically symmetric gravitational collapse in an arbitrary, albeit finite number of spacetime dimensions. In particular, we take the matter fields to belong to the general class of type-I matter fields. Choosing this class allows for a unified treatment of previous studies, since dust, tangential pressure, perfect fluid, scalar fields and Chaplygin gases \cite{PhysRevD.66.043507}, among other kinds of matter, are included in this class. Moreover, we only impose mild energy conditions (namely, the dominant and the null energy conditions) on the spacetime from physical considerations. No further assumptions, such as the existence of Killing fields (apart from those induced by spherical symmetry), or self-similarity \cite{Joshi:1992vr,PhysRevLett.59.2137,PhysRevLett.60.241} are imposed. In section \ref{sec2}, we describe the formalism employed to study spherically symmetric gravitational collapse. This includes a discussion of the \textit{interior} and \textit{exterior} spacetimes for the collapsing matter field. Here, we also construct regular initial data in terms of the matter and geometric variables under consideration. In section \ref{sec3}, we discuss necessary and sufficient conditions for the formation of a naked singularity (in any dimension). In section \ref{sec4}, we give a sufficient condition for the naked singularity to be gravitationally strong. Implications for the cosmic censorship conjectures and future directions are discussed in section \ref{sec5}. Throughout the paper (unless mentioned otherwise) we consider a spacetime to be a $N$ dimensional Lorentzian manifold, endowed with an at least $\mathcal{C}^2$ Lorentzian metric with signature $(1,N-1)$.

\section{The Mathematical Model of a Collapsing Star}\label{sec2}

\noindent We model the gravitational collapse of a spherically symmetric matter cloud by a collection of concentric spherical shells $\mathcal{S}^{N-2}_r$ of co-dimension two, identified by its radial coordinate $r$. As time commences, either all these shells contract to a point, or only shells identified by radial coordinates less than or equal to $r_e<r_b$ (where $r_b$ is the radial coordinate of the largest shell) contract to a point, or none of them contract to a point. We refer to the resultant phenomenon in the first case as \textit{unhindered gravitational collapse}, and in the remaining two cases as \textit{hindered gravitational collapse}. The point to which these shells collapse (if they do) is referred to as the \textit{singularity}. The hindered collapse of matter clouds may lead to either a singular equilibrium configuration \cite{Joshi:2011zm}, or a regular equilibrium configuration \cite{mosani2024regular}.\\

\noindent To be precise, we require two separate spacetimes: one that describes the interior collapsing cloud, and the other that describes the surrounding exterior region. These are then matched on a timelike hypersurface representing the boundary of the collapsing cloud, via an appropriate procedure \cite{Israel:1966rt}. In this fashion, we construct a ``larger" spacetime by taking the union of the interior and the exterior spacetime, and the matching hypersurface. In what follows, we briefly review the interior and exterior spacetimes for a spherically symmetric unhindered gravitational collapse. We then proceed to construct and evolve the regular initial data.
 
\subsection{The Interior Spacetime}

\noindent Let $(M,\mathbf{g})$ be a spherically symmetric $N$ dimensional spacetime with $\partial M \eqcolon B \hookrightarrow M$. We call this the \textit{interior spacetime}. Let $\gamma: \mathbb{R}\supseteq I \to M$ be the worldline of a timelike particle constituting the matter field $\mathbf{T} : M\to \mathcal{T}^0_{\ 2} M$ over $M$. For each curve parameter $\lambda \in I$, let $\textbf{v}_{\lambda}\in \mathcal{T}_{\gamma(\lambda)}M$ be the tangent vector to $\gamma$ at the point $\gamma(\lambda)\in M$. Choose a basis $\{\mathbf{e}_a\}=\{\mathbf{e}_1,\dots, \mathbf{e}_N\}$ of each tangent space $\mathcal{T}_{\gamma(\lambda)} M$, such that (i) $\forall\ \lambda \in I: \mathbf{g}_{\gamma(\lambda)}(\mathbf{e}_a, \mathbf{e}_b)=\eta_{ab}$, and (ii) $\mathbf{e}_1 \coloneqq \textbf{v}_{\lambda}$. Here, $\eta_{ab}$ are the components of the $N$ dimensional Minkowski metric. Condition (i) implies that this frame is orthonormal (in the Lorentzian sense), and condition (ii) stipulates that the particle does not move spatially with respect to itself.  Such a choice of basis is called a \textit{comoving frame}, in the sense that the frame moves with the matter. The corresponding choice of basis for the respective cotangent spaces is called the \textit{comoving coframe}.\\

\noindent By definition, a \textit{type-I matter field} admits one timelike and $N-1$ spacelike eigenvectors in a comoving frame. Note that $\mathbf{e}_1$ is timelike, since it is the tangent to a timelike geodesic. Moreover, condition (i) implies that $\mathbf{e}_2, \dots, \mathbf{e}_N$ are all spacelike. We now introduce a chart map $\mathbf{x}:M\to \mathbb{R}^N$ such that $\forall\ p \in M: \mathbf{x}(p)=(\mathbf{x}^1 (p),\dots,\mathbf{x}^N (p))\coloneqq(t,r,\theta^1,\dots,\theta^{N-2})\in \mathbb{R}^N$. Further, we require that for every $1\leq \mu \leq N$, the inverse mapping $ (\mathbf{x}^\mu)^{-1} : \mathbb{R}\supseteq I \to M $ is the integral curve of the vector field generated by $\mathbf{e}_{\mu}$ (as defined above). We call these \textit{comoving coordinates}. For future use, we define that $\forall\ 1\leq i \leq N: \mathbf{X}_i \coloneqq \text{range}(\mathbf{x}^i) \subset \mathbb{R}$, and $\mathbf{X} \coloneqq \mathbf{X}_1 \times \mathbf{X}_2 \subset \mathbb{R}^2$.\\

\noindent The comoving frame may now be denoted by the set $\{\partial_{a} : 1\leq a \leq N\}\subset \mathcal{T}_p M$, and the comoving coframe may be denoted by the set $\{d\mathbf{x}^b : 1\leq a \leq N \} \subset \mathcal{T}_p^* M$. Consequently, one can construct a basis for $\mathcal{T}^0_{\ 2, p} M$ given by $\{d\mathbf{x}^a \otimes d\mathbf{x}^b : 1\leq a,b \leq N \}$. Then, the spherically symmetric metric $\mathbf{g}: M\to \mathcal{T}^0_{\ 2} M $ can be expanded in this basis in terms of three arbitrary functions $\phi: \mathbf{X}\to \mathbb{R}$, $\psi: \mathbf{X}\to \mathbb{R}$ and $R: \mathbf{X}\to \mathbb{R}$ as follows:
\begin{equation}\label{1}
\begin{aligned}
    \mathbf{g}_p &= -e^{2\phi(t,r)}\ dt \otimes dt + e^{2\psi(t,r)}\ dr\otimes dr + \mathbf{h}_p,~\textrm{where}\\
    \mathbf{h}_p &= R^2 (t,r)  \sum_{i=1}^{N-2} \left( \prod_{j=1}^{i-1} \sin^2{\theta^j}\right) d\theta^i \otimes d\theta^i.
    \end{aligned}
\end{equation}
\noindent Here the notation $\mathbf{g}_p$ is used to denote the metric evaluated at a point $p\in M$. Similar notation shall be adhered to for any other tensor field henceforth. Likewise, one can construct a basis $\{\partial_a \otimes d\mathbf{x}^b : 1\leq a,b \leq N \}$ for $\mathcal{T}^1_{\ 1,p} M$. The energy-momentum tensor for a type-I matter field can then be expanded in this basis as follows:
\begin{equation}\label{2}
    \mathbf{T}_p = -\rho(t,r)\ \partial_t \otimes dt + p_r (t,r)\  \partial_r\otimes dr + p_\theta (t,r) \sum_{i=1}^{N-2} \partial_{\theta^i} \otimes d\theta^i. 
\end{equation}
\noindent Here, the maps $\rho: \mathbf{X} \to \mathbb{R}$, $p_r: \mathbf{X} \to \mathbb{R}$, and $p_{\theta} : \mathbf{X} \to \mathbb{R}$ are called the \textit{energy density}, \textit{radial pressure}, and \textit{tangential pressure} of the cloud undergoing gravitational collapse. Additionally, in equation (\ref{1}), the function $R: \mathbf{X} \to \mathbb{R}$ is interpreted as the \textit{physical radius} of the matter cloud. The physical significance of the metric functions $\phi$ and $\psi$ will be clarified later.

\subsection{The Exterior Spacetime}
 
\noindent Let $(\Tilde M,\mathbf{\Tilde{g}})$ be a spherically symmetric $N$ dimensional spacetime with $\partial \Tilde M \eqcolon \Tilde{B} \hookrightarrow \Tilde{M}$. We call this the \textit{exterior spacetime}. This spacetime models the region surrounding the collapsing cloud, which may be done by considering $(\Tilde M,\mathbf{\Tilde{g}})$ to be the \textit{generalized radiating Vaidya spacetime} \cite{Wang_1999}.\\
 
\noindent We introduce a different chart $\mathbf{y}:\tilde{M} \to \mathbb{R}^N$, such that $\forall\ p \in \tilde{M}: \mathbf{y}(p) =(\mathbf{y}^1,\dots,\mathbf{y}^N)(p)= (t_v, r_v, \theta_v^1, \dots, \theta_v^{N-2})\in \mathbb{R}^N$, where $t_v$ is called the \textit{retarded} (or \textit{exploding}) \textit{time}, $r_v$ is called \textit{generalized Vaidya radius}, and for each $1\leq i \leq N-2$, the numbers $\theta^i_{v}$ are the \textit{angular coordinates} for the generalized Vaidya spacetime. For future use, we define that $\forall 1\leq i \leq D: \mathbf{Y}_i \coloneqq \text{range}(\mathbf{y}^i) \subset \mathbb{R}$, and $\mathbf{Y} \coloneqq \mathbf{Y}_1 \times \mathbf{Y}_2 \subset \mathbb{R}^2$. In the corresponding coordinate basis, the exterior generalized Vaidya metric $\tilde{\mathbf{g}}: \tilde{M} \to \mathcal{T}^0_{\ 2} \tilde{M}$ is written as \cite{PhysRevD.92.024041},
 \begin{equation}\label{3}
 \begin{aligned}
\tilde{\mathbf{g}}_p &= -\left(1-\frac{2 \mathfrak{M}(t_v,r_v)}{r_v^{N-3}}\right)\ dt_v \otimes dt_v - 2\ dt_v \otimes dr_v + \tilde{\mathbf{h}}_p,~\textrm{where}  \\
\tilde{\mathbf{h}}_p &=  r_v^2  \sum_{i=1}^{N-2} \left( \prod_{j=1}^{i-1} \sin^2{\theta_v^j}\right) d\theta_v^i \otimes d\theta_v^i.
\end{aligned}
 \end{equation}
\noindent Here, $\mathfrak{M}: \mathbf{Y} \to \mathbb{R}^{+}_0$ is called the \textit{generalized Vaidya mass}. The exterior spacetime must be ``glued" to the interior spacetime in order to form one total spacetime describing the phenomenon of gravitational collapse.\\

 \noindent We construct the total spacetime $(\mathscr{M},\mathfrak{g})$ from $(M,\mathbf{g})$ and $(\Tilde{M},\Tilde{\mathbf{g}})$ as follows \cite{Mena:2012zza}:
\begin{enumerate}
    \item $\mathscr{M} \cong M \sqcup \Tilde{M}$ as Lorentzian manifolds. Moreover, for every $p\in M$, we have $\mathfrak{g}_p \coloneqq \mathbf{g}_p$ and for every $q \in \tilde{M}$, we have $\mathfrak{g}_q \coloneqq \mathbf{g}_p$.
    \item Identify points in $B$ with points in $\Tilde{B}$ such that $B \cong \tilde{B}$ as differentiable manifolds. We may then say that these boundaries are diffeomorphic to a $(N-1)$ dimensional manifold $\mathscr{S}$.
    \item The following junction conditions given by Israel \cite{Israel:1966rt} must hold: 
    \begin{equation}\label{4}
    \begin{aligned}
       \leftidx{^B}{\mathbf{g}}{} &\equiv \leftidx{^{\Tilde B}}{\Tilde{\mathbf{g}}}{}, ~\textrm{and}\\
   \mathbf{K}& \equiv \Tilde{\mathbf{K}}.
    \end{aligned}
    \end{equation}
   \noindent  Here $\leftidx{^{B}}{\mathbf{g}}{}$ denotes the metric induced on $B \hookrightarrow M$ and $\leftidx{^{\Tilde B}}{\Tilde{\mathbf{g}}}{}$ denotes the metric induced on $\tilde{B} \hookrightarrow \Tilde{M}$. Additionally, the maps $\mathbf{K} : B \to \mathcal{T}^0_{\ 2} B$ and $\tilde{\mathbf{K}}:\tilde{B}\to \mathcal{T}^0_{\ 2} \tilde{B}$ are defined pointwise in the following manner:
\begin{equation}\label{5}
    \begin{aligned}
  \forall\ p\in B;\  \mathbf{K}_p:& \mathcal{T}_pB \times \mathcal{T}_pB \to\mathbb{R}\\
  &(X,Y)\mapsto \mathbf{g}_p\left(\nabla_X Z, Y\right), ~\textrm{and}\\
\forall\ q\in \tilde{B};\ \Tilde{\mathbf{K}}_q: & \mathcal{T}_q\tilde{B} \times \mathcal{T}_q\tilde{B} \to\mathbb{R}\\
&(X,Y)\mapsto \Tilde{\mathbf{g}}_q \left(\nabla_X Z, Y\right).
\end{aligned}
\end{equation}
\noindent These are the extrinsic curvatures (or second fundamental forms) of $B \hookrightarrow M$ and $\tilde{B} \hookrightarrow \Tilde{M}$ respectively. Here $Z\in \Gamma(TN)$ is the global unit normal vector field to $\mathscr{S}$.\\
\end{enumerate}
\color{black}
\noindent \textcolor{black}{Israel's conditions imply that there are no finite or infinite jump dicontinuities in the metric of the total spacetime and the extrinsic curvatures of the matching hypersurface, as one moves from the interior to the exterior}. These conditions were used in $N=4$ case to match interior perfect fluid seeded spacetime with the exterior generalized Vaidya spacetimes in \cite{PhysRevD.76.084026}. Giambo and Quintavalle \cite{Giambo:2007ps} have generalized the procedure to $N$ dimensions, and as such we will not focus on this aspect here.
%The boundaries of the interior and exterior spacetimes $M$ and $\tilde{M}$ are given by $\Sigma_{r_b} \hookrightarrow \bar{M}$ and $\tilde{\Sigma}_{r_{v_b}} \hookrightarrow \Bar{\tilde{M}}$, such that $\Sigma_{r_b} \cong \tilde{\Sigma}_{r_{v_b}}$ as differentiable manifolds but not as pseudo-Riemannian manifolds. In particular, one may consider that $\Sigma_{r_b}$ and $\tilde{\Sigma}_{r_{v_b}}$ are diffeomorphic \cite{Mena:2012zza} to a $(D-1)$ dimensional manifold $\Sigma$. Then, using the Israel-Darmois junction conditions \cite{Israel:1966rt}, the interior and exterior spacetimes are to be ``glued" on $\Sigma$. This involves matching the first and second fundamental forms of $M$ and $\tilde{M}$ induced on $\Sigma_{r_b}$ and $\tilde{\Sigma}_{r_{v_b}}$ respectively. 

%A discussion on how this matching is to be carried out in the $N=4$ case for generalized Vaidya spacetimes is given in \cite{PhysRevD.76.084026}. Generalizing the procedure to $N$ dimensions is fairly straightforward, and we shall not focus on this aspect here. 

\subsection{Construction of Regular Initial Data}

\noindent From the dynamical point of view, the problem of gravitational collapse of a spherically symmetric matter cloud can be treated as an initial value problem, wherein regular initial data defined on an initial spacelike hypersurface $\Sigma_{t_i} \hookrightarrow M$ characterized by $t=t_i$, is evolved in the future subject to the Einstein equations. Moreover, one would like to impose some additional constraints on the initial data in terms of some physically motivated conditions. One of these restrictions is that the metric components $g_{ab} : \mathbb{R}^{N} \to \mathbb{R}$ in the chosen basis must be at least $\mathcal{C}^2$ everywhere. This is required to solve the Einstein equations in the first place. \\

\noindent Additionally, we require the stress-energy tensor to satisfy the \textit{weak energy condition} (WEC) and the \textit{null energy condition} (NEC), which are together given by the following pointwise defined condition,
\begin{equation}\label{6}
     \forall\ \mathsf{V} \in \mathcal{T}_p M : ~\lVert \mathsf{V} \rVert_p \leq 0 \implies \mathbf{T}_p(\mathsf{V},\mathsf{V}) \geq 0.
\end{equation}
\noindent Here, $\lVert \mathsf{V} \rVert_p \coloneqq \mathbf{g}_p (\mathsf{V},\mathsf{V})$. Moreover, we also ensure the satisfaction of \textit{dominant energy condition} (DEC), which requires that the WEC holds in addition to the following pointwise defined condition:
\begin{equation}\label{7}
    \forall\ \mathsf{V} \in \mathcal{T}_p M: ~\lVert \mathsf{V} \rVert_p <0 \implies \lVert (\mathbf{T}_p(\mathsf{V},\cdot))^{\sharp_p} \rVert_p \leq 0.
\end{equation}
\noindent Here, $(\cdot)^{\sharp_p}$ denotes the pointwise evaluation of the globally defined musical isomorphism \cite{sasane2022mathematical} $\sharp: \mathcal{T}^* M \to\mathcal{T}M$. These energy conditions are rather mild, physically motivated conditions, conveying the idea that the energy density along light trajectories, and as measured by an observer locally, is non-negative, in addition to the notion of causal matter flow. Furthermore, no other energy conditions are imposed. For the specific form of energy-momentum tensor given in equation (\ref{2}), the DEC and NEC are together equivalent to the satisfaction of the following condition:
\begin{equation}\label{8}
\begin{aligned}
    \forall\ &(t,r) \in \mathbf{X}: \Delta_{(t,r)} \subset \mathbb{R}^{+}_0, ~\textrm{where}  \\
   \Delta_{(t,r)} &\coloneqq \{  \rho (t,r), ~(\rho + p_r)(t,r), ~(\rho + p_\theta) (t,r), ~(\rho-|p_r|) (t,r), ~(\rho - |p_\theta|)(t,r) \}.
\end{aligned}
\end{equation}

\noindent Now we define a particular notion of mass for the spacetimes under consideration. As we shall see later, this mass function can be used to characterize the apparent horizon.\\

\begin{definition}\label{d1}
    (Misner-Sharp mass function, Giambo et. al. \cite{Giambo:2002tp})\\
    \noindent For a $N-$dimensional spherically symmetric spacetime $(M, \mathbf{g})$, with $\mathbf{g}$ given by equation (\ref{1}), the Misner-Sharp mass function is a map $F: \mathbf{X} \to \mathbb{R}$ defined by,
    \begin{equation}\label{9}
        F(t,r) \coloneqq R^{N-3}(t,r) (1-\mathbf{g}^{\sharp}_{\ p}(d R, d R)) (t,r).
    \end{equation}

    \noindent Here, $\forall\ p \in M: \mathbf{g}^{\sharp}_{\ p} : \mathcal{T}^*_p M \times \mathcal{T}^*_p M \to \mathbb{R}$ is the scalar product on each cotangent space, and will be referred to as the inverse metric.\\
\end{definition}

\begin{remark}\label{r1}
    \noindent In the context of the Hamiltonian formulation of relativity, one can define various notions of mass \cite{Faraoni:2020mdf}. For instance, one has the Hawking quasi-local mass \cite{Hawking:1968qt}, the Arnowitt-Deser-Misner (ADM) mass \cite{Arnowitt:1962hi} and the Bondi-Sachs (BS) mass \cite{Szabados:2009eka}. The Misner-Sharp mass is equivalent to the Hawking quasi-local mass under the assumption of spherical symmetry. Moreover, the ADM mass and the BS mass are equivalent to the Misner-Sharp mass in the limit of approach to spatial infinity and null infinity respectively \cite{PhysRevD.49.831}.\\
\end{remark}

\noindent With these definitions in mind, for the metric given in equation (\ref{1}), the Einstein equations can be written in the following form:
\begin{subequations}
    \begin{align}
        \rho(t,r) &= \frac{(N-2) F'(t,r)}{2 R^{N-2}(t,r) R'(t,r)}, \label{10a}\\
        p_r(t,r) &= - \frac{(N-2)\dot{F}(t,r)}{2 R^{N-2}(t,r) \dot{R}(t,r)},\label{10b}\\
        \phi'(t,r) &=\frac{(N-2)(p_\theta (t,r)-p_r (t,r))}{\rho(t,r)+p_r (t,r)}\frac{R'(t,r)}{R(t,r)} - \frac{p'_r (t,r)}{\rho(t,r)+p_r (t,r)},\label{10c}\\
         \dot{R}'(t,r) &=\frac{1}{2}\left(\dot{R}(t,r) \frac{H'(t,r)}{H(t,r)} + R'(t,r)\frac{\dot{G}(t,r)}{G(t,r)}\right). \label{10d}
    \end{align}
\end{subequations}
\noindent Here, we have used the notations,
\begin{equation}\label{11}
        F'(t,r) = \frac{\partial F}{\partial r}\bigg|_{t} (t,r); \ \ \ \dot{F}(t,r) = \frac{\partial F}{\partial t}\bigg|_{r} (t,r).
\end{equation}
\noindent The same notations applies to the physical radius $R$. Moreover, the functions $G: \mathbf{X}\to \mathbb{R}$ and $H: \mathbf{X}\to \mathbb{R}$ are defined as,
\begin{equation}\label{12}
    G(t,r) \coloneqq e^{-2\psi(t,r)} (R')^2(t,r); \ \ \ H(t,r) \coloneqq e^{-2\phi(t,r)} (\dot{R})^2 (t,r).
\end{equation}
\noindent In terms of these functions, equation (\ref{9}) gives,
\begin{equation}\label{13}
    F(t,r) = R^{N-3} (t,r) (1-G(t,r) + H(t,r)).
\end{equation}
\noindent In the context of gravitational collapse of a dust cloud, the Misner sharp mass function is to be interpreted as the quasi-local mass contained inside a spherical matter shell of radius $r$ at time $t$. In fact, from equation (\ref{10a}), we find that at a given time $t$,
\begin{equation}\label{14}
    F(t,r) = \frac{2}{N-2}\int_0^{R}\ d\bar{R}\ \bar{R}^{N-2} \rho.
\end{equation}
\noindent It is possible for the Misner-Sharp mass function to assume negative values a priori. However, in the context of the positive energy theorem(s) of classical relativity \cite{Schoen:1979zz,Witten:1981mf}, the Misner-Sharp mass corresponds to the mass of the spacetime. As a consequence, it is necessary that the Misner-Sharp mass function assume only non-negative values, i.e., 
\begin{equation}\label{15}
    \text{range}(F) \subseteq \mathbb{R}^+_0.
\end{equation}
 \noindent As we shall see later, the Misner-Sharp mass function can be used to characterize the apparent horizon, much like the Schwarzschild mass can be used to characterize the event horizon, which is equivalent to the apparent horizon for the Schwarzschild spacetime. In fact, the Schwarzschild mass is equivalent to the Misner-Sharp mass function for the Schwarzschild spacetime. Using equations (\ref{6}) and (\ref{10a}), the WEC and NEC can be formulated along any causal curve as,
\begin{equation}\label{16}
     \forall\ (t,r) \in \mathbf{X}: ~\{ F'(t,r), ~R'(t,r) \} \subset \mathbb{R}^{+}_0.
\end{equation}
\noindent Some restrictions on the behaviour of the physical radius $R(t,r)$ are in order. The behaviour of $\dot{R}:\mathbf{X} \to \mathbb{R}$ determines whether the matter cloud undergoes unhindered gravitational collapse, hindered gravitational collapse, or expands. Clearly, for the matter field to undergo unhindered gravitational collapse, we must have,
\begin{equation}\label{17}
   \text{range}(\dot{R})\subseteq \mathbb{R}^-.
\end{equation}
\noindent If one has $\dot{R}(t_f,r) \geq 0$ at some later time $t=t_f$, the collapsing matter would either halt or bounce back. We do not concern ourselves with such cases. Instead, we focus on the scenario wherein all the matter shells completely collapse to a singularity. Now, for reasons to be specified later, we introduce a scaling function for the physical radius.\\

\begin{definition}\label{d2}
    (Scaling function)\\
    The scaling function is a map $v: \mathbf{X}\to (0,1]$ defined by,
 \begin{equation}\label{18}
    v(t,r) \coloneqq \frac{R(t,r)}{r}.
\end{equation}
\end{definition}

\begin{remark}\label{r2}
    The scaling function $v(t,r)$ is a monotone decreasing function of $t$ throughout the gravitational collapse. This allows us to use $v$ as a time-coordinate. In other words, we employ the coordinate transformation $t \mapsto v(t,r)$ for a fixed $r$. For the remainder of the paper, while referring to $v$ as a coordinate, we shall (abusively) drop the argument $(t,r)$. Moreover, a function $f(t,r)$ in the $(t,r)$ plane will be denoted in the $(v,r)$ plane as a corresponding transformed function $\tilde{f}(v,r)$.\\
\end{remark}

\noindent We now recall some useful definitions from the theory of spherically symmetric gravitational collapse.\\

\begin{definition}\label{d3}
    (Epoch)\\
    Let $(M,\mathbf{g})$ be the interior spacetime. A spacelike hypersurface $\Sigma \subset M$ with normal $\partial_t$, defined by some $t=t_1 \in \mathbf{X}_1$, or equivalently by some $v=v_1 \in (0,1]$ is called a (regular) epoch.\\
\end{definition}

\begin{definition}\label{d4}
    (Time curve) \\
    A map $\tilde{t} : (0,1] \times \mathbf{X}_2 \to \mathbb{R}$ describing the time $\tilde{t}(v,r)$ for a matter shell identified by a radial coordinate $r$ to reach the epoch defined by a constant $v\in(0,1]$ is called the time curve.\\
\end{definition}

\begin{definition}\label{d5}
    (Singularity curve)\\
A singularity curve is a map $t_s: \mathbf{X}_2 \to \mathbb{R}$ defined in terms of the time curve $\tilde{t}$, as,
    \begin{equation}\label{19}
        \forall r\in \mathbf{X}_2 : 
        t_s(r) \coloneqq \lim_{v\to 0} \tilde{t}.
    \end{equation}
\end{definition} 

\noindent Clearly, the singularity curve describes the time taken by a matter shell identified by a radial coordinate $r$ to approach the \textit{singular epoch}, i.e., the epoch defined in the limit $v\to 0$, or equivalently $t \to t_s (r)$.\\

\begin{definition}\label{d6}
    (Shell-crossing singularity, Szekeres and Lun \cite{Szekeres:1995gy})\\
    A spacetime $(M,\mathbf{g})$ is said to admit a shell-crossing singularity if either of the following conditions hold:
    \begin{enumerate}[(i)]
        \item Let $\gamma: [\lambda_0,0) \to M$ be an incomplete causal curve which approaches the singularity as the curve parameter $\lambda \to 0$. If $J: [\lambda_0,0) \to \mathcal{T}M$ such that for every $\lambda \in [\lambda_0,0)$, $ J(\lambda) \in \mathcal{T}_{\gamma(\lambda)}M$ constitutes a Jacobi field along $\gamma$, then in a parallely propagated (pp) orthonormal frame along $\gamma$, its components $J^a \in \mathcal{C}^k (M)$ have finite non-zero values as $\lambda \to 0$.
        \item It can be covered by a set of $\mathcal{C}^1$ regular boundary points, i.e., it is removable by a $\mathcal{C}^1$ chart transition.\\
    \end{enumerate}
\end{definition}

\noindent The second part of the above definition employs the abstract boundary formulation of Scott and Szekeres \cite{Scott:1994sn} to identify singularities as boundary points attached to spacetime. \\
\begin{definition}\label{d7}
    (Shell-focussing singularity)\\
   A spherical matter shell identified by a radial coordinate $r\in \mathbf{X}_2$ is said to terminate at a shell-focussing singularity if its physical radius $R$ vanishes in the limit of approach to the singular epoch, i.e., if,
   \begin{equation}\label{20}
       \lim_{t\to t_s (r) } R = 0.
   \end{equation}
\end{definition}
\noindent From equation (\ref{10a}), it is seen that the density $\rho$ diverges in the limit $R(t,r) \to 0$ and/or $R'(t,r) \to 0$. Clearly, $R(t,r) = 0$ indicates the vanishing of all matter shells, and $R'(t,r) \geq 0$ implies a crossing or intersection of distinct matter shells. It has been shown previously \cite{PhysRevD.47.5357,Nolan_2003} that the spacetime can be extended through the singularities characterized by $R'(t,r) \to 0$ and $R(t,r) \to c\in \mathbb{R}$. Hence, such singularities can be classified as shell-crossing following definition \ref{d6}. On the other hand, the singularities characterized by $R(t,r) \to 0$ and $R'(t,r) \to 0$ are simultaneously shell-crossing and shell-focussing following definitions \ref{d6} and \ref{d7} respectively. In a physically sound gravitational collapse scenario, one would expect that matter shells identified by smaller values of $r$ collapse prior to those identified by larger values of $r$. Hence we would like to avoid the formation of shell-crossing singularities before the final shell-focussing singularity by construction. The behaviour of $R'$ is thereby restricted as,
\begin{equation}\label{21}
    \text{range}(R') \subseteq \mathbb{R}^+_0.
\end{equation}
\noindent This condition ensures that the physical radius cannot exhibit a non-linear response to the comoving radius, on the initial hypersurface.\\

\begin{lemma}\label{l1}
    \textit{Let $\beta\in \mathbb{R}$. For a complete gravitational collapse without initial shell-crossings, from a regular initial hypersurface $\Sigma_{t_i}$ subject to the Einstein equations, we have,}
\begin{equation}\label{22}
    R(t_i,r) = r^\beta \implies \beta = 1.
\end{equation}
\end{lemma}
\begin{proof}
    We have $R'(t_i,r)= \beta r^{\beta -1}$. Suppose $\beta \neq 1$. Then, at $t=t_i$, either $R'(t_i,0) = 0$ or $R'(t_i,0)\to \infty$. The former contradicts the assumption that there are no initial shell-crossings, and the latter contradicts the requirement that the metric functions must be at least $\mathcal{C}^2$ everywhere initially in order to solve the Einstein equations. Hence we must have $\beta = 1$.\\
\end{proof} 

\noindent Using this lemma and the initial scaling freedom for the physical radius, we make the choice $R(t_i,r) \coloneqq r$. We now recall a useful theorem regarding the occurrence of shell-crossing singularities in the context of spherically symmetric gravitational collapse. \\

\begin{theorem}\label{t1}
(Joshi and Saryakar \cite{Joshi:2012ak})\\
Given $\mathbb{R}\ni\epsilon>0$, there exists a number $r_1\in \mathbb{R}^+$ such that for every $r\in[0,r_1]$, we have $R'(t,r) >0$ upto $v(t,r)=\epsilon$.\\
\end{theorem}

\noindent This theorem says that given any spherically symmetric unhindered gravitational collapse model,  there always exists a $r_1\in \mathbb{R}^+$ such that there are no shell-crossings present for every $r\in[0,r_1]$, throughout the evolution of collapse upto any arbitrarily close epoch to the final singularity. With this in hindsight, we see that lemma \ref{l1} is a corollary of theorem \ref{t1}.\\

\noindent As an upshot of our discussion thus far, we find that the following conditions on the scaling function describe the spherically symmetric gravitational collapse of matter fields. 
\begin{equation}\label{23}
    v(t_i,r) = 1; \ \ \ \ \lim_{t\to t_s(r)} v = 0; \ \ \ \  \dot{v}(t,r) < 0. 
\end{equation}
\noindent Clearly, $R(t,r)$ vanishes at $r=0$ on all regular epochs. However, $R(t,r)$ also vanishes in the limit of approach to the singular epoch $t\to t_s (r)$. These cases are to be distinguished by choosing a suitable form of the Misner-Sharp mass function $F(t,r)$, such that it vanishes sufficiently fast in the limit $r\to 0$, so that even if $R(t,r)\to 0$ at the center of all regular epochs, the density and curvatures remain finite. On the other hand, by introducing $v(t,r)$ as given by equation (\ref{18}) and (\ref{23}), we get a better characterization of the shell-focussing singularity, namely $v\to 0$ as $t\to t_s(r)$. Note, however that unlike $R(t_i,r)$, we have $v(t_i,r) =1$. Hence the regular and singular epochs are appropriately distinguished. \\

\noindent Moving further, from equation (\ref{10a}), it is seen that the regularity and finiteness of $\rho$ requires that $F'(t_i,r) \sim \mathcal{O}(r^{N-2})$. Hence, the regularity of the initial density profile requires that the Misner-Sharp mass function must have the following general form:
\begin{equation}  \label{24}
    F(t,r) = r^{N-1} \tilde{\mathcal{M}}(v,r).
\end{equation}
\noindent Here, ${\tilde{\mathcal{M}}}: (0,1] \times \mathbf{X}_2 \to \mathbb{R}^{+}$ is at least $\mathcal{C}^1$ for $r=0$ and at least $\mathcal{C}^2$ for $r\in \mathbb{R}^{+}$. Such construction implies that the Misner-Sharp mass vanishes at the centre of the collapsing cloud on all regular epochs, as well as in the limit of approach to the singular epoch, i.e., 
\begin{equation}\label{25}
   \forall\ t \in \mathbb{R}\setminus \{t_{s_0}\}: ~F(t,0) = 0\ ;\ \ \ \ \ \lim_{(t,r) \to (t_{s_0},0)} F = 0. 
\end{equation}
\noindent Here, $t_{s_0} \coloneqq t_s(0)$ and hence $(t,r)\to (t_{s_0},0) \iff (v,r)\to(0,0)$ represents the shell-focussing singularity at the centre, which we shall call the \textit{central singularity}. Since the Misner-Sharp mass vanishes at such singularities, they are sometimes also referred to as \textit{massless singularities} \cite{PhysRevD.47.5357}. As we shall see later, for spherically symmetric gravitational collapse of type-I matter fields only the class of central shell-focussing singularities can be naked. On the other hand, the class of non-central shell-focussing singularities will turn out to be covered by an apparent horizon. \\

\noindent We have now defined on the initial hypersurface $\Sigma_{t_i}$, seven functions as initial data, given by,
\begin{equation}\label{26}
\begin{aligned}
     \phi(t_i,r) &= \phi_0 (r); \ \ \ \psi(t_i,r) = \psi_0(r);\ \ \ R(t_i,r) = r ;\ \ \ F(t_i,r) = r^{N-1} \tilde{\mathcal{M}}(1,r);\\
 \rho(t_i,r) &= \rho_0(r);\ \ \ p_r (t_i,r) = p_{r_0} (r); \ \ \ p_\theta (t,r) = p_{\theta_0} (r). 
\end{aligned}
\end{equation}
\noindent Sufficient conditions to maintain the regularity of $\phi_0 (r)$ are given in \cite{PhysRevD.76.084026}. These sufficient conditions are motivated from physical considerations such as the vanishing of pressure gradient forces and the vanishing of the difference between radial and tangential pressures near the center of the collapsing matter cloud (i.e., the matter behaves like a perfect fluid near its center). These conditions, in addition to those discussed in this subsection, guarantee the regularity of the metric components and the matter density and pressures at the initial epoch.

\subsection{Singular Solutions from Regular Initial Data}

\noindent Having prescribed regular initial data on the spacelike hypersurface $\Sigma_{t_i} \subset M$, we now wish to investigate its future development $(M, \mathbf{g})$. As aforementioned, the complete gravitational collapse of a matter cloud may result in either a black hole or a naked singularity. Our approach here will be to investigate the existence and finiteness of tangents to a non-zero measure set of outgoing causal geodesics in the limit of approach to the shell-focussing singularity. As we shall see, the polarity of their tangents will decide whether the matter cloud collapses into a black hole or a naked singularity.\\

\noindent \noindent It is clear from the Einstein equations (\ref{10a}) to (\ref{10d}) and equation (\ref{11}), that not all of the initial data is independent of each other. Indeed, we have five equations and seven unknowns. One possible approach in such a scenario is to introduce an appropriate equation of state relating the density and pressure profiles \cite{10.1093/mnras/staa1493,Ahmad_2018,Goswami:2004kq}. However, in order not to bring in undue assumptions, we will not fix any equation of state for our analysis. Instead, the formalism employed here allows one to work with and construct general classes of functions that are required to have a certain form by the regularity conditions, much like the form taken by the Misner-Sharp mass function, as dictated by the regularity of the energy density. To begin with, we define a suitably differentiable function $\tilde{A}: (0,1] \times \mathbf{X}_2 \to \mathbb{R}$ such that,
\begin{equation}\label{27}
    \tilde{A},_{v}(v,r) \coloneqq \frac{\partial \tilde{A}}{\partial v} \bigg|_r (t,r) = r\frac{\phi'(t,r)}{R'(t,r)}.
\end{equation}
\noindent The regularity of $\phi':\mathbf{X} \to \mathbb{R}$ will decide the regularity of $A_{,v} : (0,1] \times \mathbf{X}_2 \to \mathbb{R}$. For any function in the $(v,r)$ plane, ($,_v$) denotes its partial derivative with respect to $v$ along an $r$ constant hypersurface and ($,_r$) denotes its partial derivative with respect to $r$ along a $v$ constant hypersurface. Using equation (\ref{27}) in equation (\ref{10d}) and integrating, yields the expression for $\tilde{G}: (0,1] \times \mathbf{X}_2$ as,
\begin{equation}\label{28}
    \tilde{G}(v,r)=b(r) e^{2\tilde{A}(v,r)}.
\end{equation}
\noindent Here $b(r)$ is the integration constant and is called the \textit{velocity profile} \cite{PhysRevD.101.044052} of the collapsing matter shell identified by $r$. The regularity of $\dot{v}: \mathbf{X} \to \mathbb{R}$ at $r=0$ implies that $b(r)$ must have the following form:
\begin{equation}\label{29}
    b(r) = 1+r^2 b_0 (r).
\end{equation}
\noindent $b_0 (r)$ is called the \textit{energy distribution} \cite{PhysRevD.101.044052} within a shell of radius $r$. Equations (\ref{21}), (\ref{27}), (\ref{28}) and (\ref{13}) yield,
\begin{equation}\label{30}
    v^{(N-3)/2}(t,r)\ \dot{v}(t,r) = -e^{\phi(t,r)} \sqrt{\tilde{\mathcal{M}}(v,r) + \frac{v^{N-3}(t,r)(b(r)e^{2\tilde{A}(v,r)}-1)}{r^2}}.
\end{equation}
\noindent Integrating equation (\ref{30}) with respect to $v$ yields the following expression for the time curve:
\begin{equation}\label{31}
\tilde{t}(v,r) = \int_{v}^1 d\bar{v} \frac{e^{-\phi(t,r)}}{\sqrt{\dfrac{\tilde{\mathcal{M}}(\bar{v},r)}{\bar{v}^{N-3}}+\dfrac{b(r)e^{2\tilde{A}(\bar{v},r)}-1}{r^2}}}. 
\end{equation}
\noindent The time curve describes the time taken by a shell labelled by the radial coordinate $r$ to arrive at some epoch defined by constant $v$ from the initial epoch $v=1$. Hence, the time taken by the shell to collapse into the shell-focussing singularity defined by $v\to 0$, i.e., the expression for the singularity curve $t_s : \mathbf{X}_2 \to \mathbb{R}$, is given by,
\begin{equation}\label{32}
t_s(r) \coloneqq \lim_{v\to 0} \tilde{t} = \int_0^1   d\bar{v} \frac{e^{-\phi(t,r)}}{\sqrt{\dfrac{\tilde{\mathcal{M}}(\bar{v},r)}{\bar{v}^{N-3}}+\dfrac{b(r)e^{2\tilde{A}(\bar{v},r)}-1}{r^2}}}. 
\end{equation}
\noindent The existence and polarity of finite tangents $(dt_s /dr)(r)$ of the singularity curve will decide the existence or otherwise of ingoing and outgoing causal geodesics approaching or emanating from the shell-focussing singularity. The existence of such causal curves in turn play a crucial role in deciding whether the corresponding singularities are naked or not. A sufficient condition for such a singularity curve to exist is the existence of an at least $\mathcal{C}^2$ time curve. From equations (\ref{31}) and (\ref{32}), it can be seen that the (at least) $\mathcal{C}^2$ regularity of the functions $\tilde{A}$, $\tilde{\mathcal{M}}$ and $\phi$ is enough to show that the time curve and hence the singularity curve terminate at the shell-focussing singularity with a finite tangent. On the other hand, even if these functions blow up in the limit of the singularity, the time curve remains finite nonetheless. Therefore, regardless of the regularity of the integrand in equation (\ref{31}), the time curve is finite. The reader is referred to the monograph \cite{joshi2007gravitational} for a detailed discussion regarding the conditions for the existence of a $\mathcal{C}^2$ time curve.\\

\noindent The existence of a regular time curve with a finite tangent in the limit of approach to the singularity implies that the collapse terminates into a singularity in a finite amount of time, which is expected from any physically reasonable collapse scenario. Equations (\ref{10a})-(\ref{10c}) along with equations (\ref{24}) and (\ref{27}) yield
\begin{subequations}
\begin{align}
    \tilde{\rho}(v,r) &= \frac{N-2}{2}\left( \frac{(N-1) \tilde{\mathcal{M}}(v,r) + r(\tilde{\mathcal{M}}_{,r} (v,r) +\tilde{\mathcal{M}}_{,v} (v,r) v'(t,r))}{v^{N-2} (v+r v'(t,r))}\right), \label{33a} \\ 
    \tilde{p}_r (v,r) &= -\frac{N-2}{2}\left(\frac{\tilde{\mathcal{M}}_{,v} (v,r)}{v^{N-2}}\right), \label{33b}\\ 
    \tilde{p}_\theta (v,r) &= p_r (v,r) + \frac{R(t,r)}{N-2}\left(\frac{p'_r (t,r)}{R'(t,r)} + \tilde{A}_{,v}(v,r) (\rho(v,r)+p_r(v,r))\right). \label{33c}
\end{align}
\end{subequations}
\noindent It is clear from these equations that the densities and pressures diverge in the limit of approach to the shell-focussing singularity, i.e.,
\begin{equation}\label{34}
    \noindent \lim_{v\to 0} \tilde{\rho} = \infty; \ \ \ \lim_{v\to 0} \tilde{p}_r = \infty;\ \ \ \lim_{v\to 0} \tilde{p}_\theta  = \infty.
\end{equation}
\noindent Hence we have the following well-known result which we write for the sake of completion: \textit{Given spherically symmetric regular initial data for a collapsing type-I matter field, on a spacelike hypersurface $\Sigma_{t_i}$ subject to the DEC, NEC and the Einstein equations, its future development\footnote{Note that we are not using the usual notion of Cauchy developments here, since we do not require the developments to be globally hyperbolic. See \cite{chruściel1991uniqueness} for a definition of general developments.} $(M, \mathbf{g})$ is singular, i.e., there exists at least one incomplete and inextendible curve $\gamma: \mathbb{R}\supseteq I \to M $}. Hence the interior spacetime $(M,\mathbf{g})$ is singular. No comments have been made on the nakedness (or otherwise) of the singularity yet.

\section{Classification of Singular End States}\label{sec3}

\noindent Now that it is guaranteed that the spherically symmetric gravitational collapse of a general type-I matter field must result in either a black hole or a naked singularity, we wish to tackle the problem of finding the conditions for the existence of each of these possible end states. More importantly, we wish to investigate the conditions under which the singularity will be naked. This involves studying trapped surfaces and the apparent horizon.\\ 

\noindent For the sake of completeness, we review some relevant notions from Lorentzian geometry that aid the study of the causal structure of a spacetime.\\

\noindent \begin{definition}
    (Geometric equivalence of curves)\\
    Let $(M,\mathbf{g})$ be a spacetime, and $U\subseteq M$ be open in $M$. Consider the set of geodesics $\Gamma_U \coloneqq \{ \gamma: \mathbb{R}\supseteq I \to U \subseteq M\}$. On this set, we define the geometric equivalence relation: for $\gamma_1$ and $\gamma_2$ $\in \Gamma_U$,  $\gamma_1 \sim \gamma_2 \iff \mathrm{Im}(\gamma_1) = \mathrm{Im}(\gamma_2)$. Two curves related via this relation are called geometrically equivalent.\\
\end{definition}

\noindent \begin{definition} 
(Geodesic congruence)\\
Under the geometric equivalence relation, we consider the quotient space $\Gamma_U/{\sim} \coloneqq \{[\gamma] | \gamma \in \Gamma_U\}$. A geodesic congruence on $U$ is a subset $C \subseteq \Gamma_U /{\sim}$ with the property that there exists a unique $[\gamma] \in C$ for each point $p\in U$. \\
\end{definition}

\noindent The geometric equivalence relation conveys the idea that curves which are reparameterizations of each other, are to be considered as the same curve. Through each point $p\in U$, there exists a unique curve $\gamma \in \Gamma_U$ through $p$, up to an affine reparametrization, such that the curve is part of the congruence. Given a spacelike hypersurface $\Sigma \subset M$, we can define two kinds of geodesic congruences through it.\\

\begin{definition}\label{d9}
    (Ingoing and outgoing null geodesic congruence)\\
   Given a spacetime $(M, \mathbf{g})$, let $U\subseteq M$ be open in $M$, and let $\Sigma\subset M$ be a spacelike hypersurface. Let $\mathsf{T}$ be a ($N-2$) dimensional closed (compact and without boundary) submanifold of $M$ such that $\mathsf{T}\subset \Sigma$. Let $\mathcal{N} \in \Gamma(\mathcal{T}\Sigma)$ be the globally defined unit normal vector field to $\mathsf{T}$. Further, define a set,
    \begin{equation}\label{35}
      \mathcal{T}Q \coloneqq \bigcup_{p\in \mathsf{T}} \mathcal{T}_p M.
    \end{equation}
    \noindent Then, a future-directed vector field $\mathcal{N}^- \in \Gamma(\mathcal{T}Q) $ satisfying $\forall p \in \mathsf{T}: \mathbf{g}_p(\mathcal{N}(p),\mathcal{N}^-(p)) = -1$ is called the ingoing null normal field. Similarly, a future-directed vector field $\mathcal{N}^+ \in \Gamma(\mathcal{T}Q)$ satisfying $\forall p \in \mathsf{T}: \mathbf{g}_p(\mathcal{N}(p),\mathcal{N}^+(p)) = 1$ is called the outgoing null normal field. The integral curves of the ingoing null normal field form the ingoing null geodesic congruence (INGC), and the integral curves of the outgoing null normal field form the outgoing null geodesic congruence (ONGC).\\
\end{definition}

\noindent The reader is referred to fig. \ref{fig2} for a visualization of this concept. Once we have defined a geodesic congruence, we can discuss its behaviour in terms of three parameters - \textit{expansion}, \textit{vorticity} and \textit{shear} \cite{Wald:1984rg}. For the purpose of this paper, we consider only hypersurface orthogonal congruences, which means their vorticity vanishes \cite{Wald:1984rg}. The expansion of a geodesic congruence can be used to define trapped surfaces and the apparent horizon. \\

\begin{figure}[t]
 			\begin{center}
				\includegraphics[width=130mm,scale=0.5]{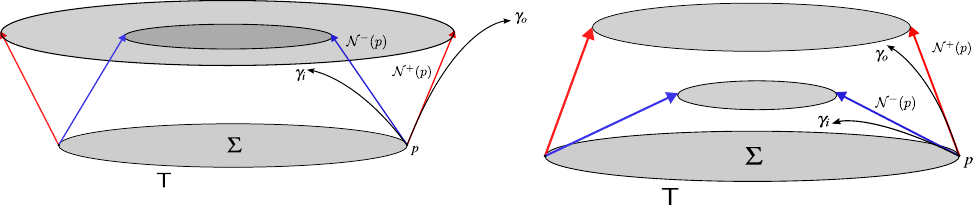}
			\end{center}
   \vspace{1em}
       		\caption{ On the LHS, we have cartoon depicting the concepts of ingoing and outgoing geodesic congruences. At any point $p\in M$, we may define two vectors $\mathcal{N}^{-} (p)$ and $\mathcal{N}^+ (p)$, that form the tangents to ingoing and outgoing curves $\gamma_i$ and $\gamma_o$ through that point. In the theory of gravitational collapse, the terminology ``ingoing" and ``outgoing" is often used for geodesics running into or emanating from a singularity. In the latter case, we refer to the singularity as being naked. On the RHS, we have a cartoon depicting a trapped surface. For such surfaces, ingoing and outgoing geodesics converge in the future (or past).} \label{fig2}
\end{figure}

\begin{definition}
    (Null mean curvature/ expansion)\\
    Let the conditions in definition \ref{d9} hold. Let $\leftidx{^\mathsf{T}}{\mathbf{g}}{}$ be the metric induced on $\mathsf{T}$ from $M$, and for each $p\in \mathsf{T}$, let $\{\mathbf{e}_1,\dots \mathbf{e}_{N-2}\}$ be an orthonormal basis for $\mathcal{T}_p \mathsf{T}$. Then, the ingoing and outgoing null mean curvatures $\Theta^-$ and $\Theta^+$ of $\mathsf{T}$ at $p\in \mathsf{T}$ are defined as,
\begin{equation}\label{36}
    \Theta^{\mp}(p) \coloneqq \mathsf{tr}_{\leftidx{^\mathsf{T}}{\mathbf{g}}{}}(\nabla \mathcal{N}^{\mp}) (p) \coloneqq \sum_{i=1}^{N-2} \leftidx{^\mathsf{T}}{\mathbf{g}}{}_p (\nabla_{\mathbf{e}_i}\mathcal{N}^{\mp}(p),\mathbf{e}_i).
\end{equation}
Moreover, since $\mathcal{N}^{-}$ and $\mathcal{N}^{+}$ are the tangent fields to null geodesics, this is also called the expansion of a null geodesic congruence through $\mathsf{T}$.\\
\end{definition} 

\begin{definition}
   (Trapped surface, marginally trapped surface, MOTS)\\
    Let $\mathsf{T}$ be defined as in definition \ref{d9}. Then,
   \begin{enumerate}[(i)]
   \item    $\mathsf{T}$ is called a trapped surface if  $\ \forall\ p \in \mathsf{T}: \Theta^{\mp}(p) <0$ (see fig. \ref{fig2}).
   \item   $\mathsf{T}$ is called a marginally trapped surface if  $\ \forall\ p \in \mathsf{T}: \Theta^{\mp}(p) \leq 0$.
    \item   $\mathsf{T}$ is called a marginally outer trapped surface (MOTS) if  $\ \forall\ p \in \mathsf{T}: \Theta^{+}(p) = 0.$\\
   \end{enumerate}
\end{definition}

\begin{definition}
  (Edge, partial Cauchy surface, trapped region, apparent horizon)\\ 
   The edge of a spacelike hypersurface $\Sigma \subset M$ is the set of points $p \in \bar{\Sigma} =\Sigma \cup \partial \Sigma$ such that for every open neighbourhood $U \subset \bar{\Sigma}$ of $p$, there exists a causal curve $\gamma: J^{-}(p) \to J^{+}(p)$ that does not intersect $\Sigma$. Further, $\Sigma$ is called a partial Cauchy surface if its edge is empty. Moreover, if $\Sigma$ is a partial Cauchy surface, then a trapped region is a $(N-1)$ dimensional submanifold $\mathscr{T} \subset \Sigma$ whose boundary $\mathscr{A}=\partial \mathscr{T}$ is a MOTS with $\Theta^+ (p) =0$ for every $p \in \mathscr{A}$. The boundary $\mathscr{A}$ is called an apparent horizon. \\
\end{definition}

\noindent For the general spherically symmetric metric $\mathbf{g}$ given by equation (\ref{1}), the expansion for an ONGC can be written (in the comoving coordinates) in terms of the metric components \cite{PhysRevD.109.064019}, using an orthonormal null tetrad as,
\begin{equation}\label{37}
   \forall\ (t,r) \in \mathbf{X}: ~\Theta^{+} (t,r) = \frac{N-2}{R^{N-3}(t,r)}\left(1-\sqrt{\frac{F(t,r)}{R^{N-3}(t,r)}}\right).
\end{equation}
\noindent Hence for a general spherically symmetric metric, the apparent horizon is characterized by the following condition:
\begin{equation}\label{38}
   \forall\ (t,r) \in \mathbf{X}: ~\Theta^{+} (t,r)= 0 \iff F(t,r) = R^{N-3} (t,r).
\end{equation}
\noindent Here $(t,r)$ represents the image of some point $p\in M$ under the comoving chart maps.

\subsection{Conditions for the Formation of Black Holes}

\noindent The apparent horizon has a direct relation with geometric wave optics. This can be seen by realizing that equations (\ref{9}) and (\ref{38}) together require that the eikonal equation $\mathbf{g}^{\sharp}_{\ p}(d R, d R) = 0$, for the general spherically symmetric metric \cite{Altas_2022,Dotti:2023elh} given by equation (\ref{1}) holds. Hence, at the apparent horizon, the gradient of the physical radius is a null covector. If the shell-focussing singularity is formed after the apparent horizon, it ends up being covered by the apparent horizon, and hence in this case, the matter cloud collapses (in finite time) to a black hole. Using equation (\ref{24}), we find the epoch $v=v_{\text{ah}}(r)$ corresponding to the formation of the apparent horizon, is a root of the following equation:
\begin{equation}\label{39}
    v^{N-3} (t,r) = r^2 \tilde{\mathcal{M}}(v,r).
\end{equation}
\noindent Using the time curve given by equation (\ref{31}), we get the expression for the \textit{apparent horizon curve} $t_{\text{ah}}: \mathbf{X}_2 \to \mathbb{R}$ as,
\begin{equation}\label{40}
    t_{\text{ah}}(r) \coloneqq t_s (r) - \int_{0}^{v_{\text{ah}}(r)} d\bar{v} \frac{e^{-\phi(t,r)}}{\sqrt{\dfrac{\tilde{\mathcal{M}}(\bar{v},r)}{\bar{v}^{N-3}}+\dfrac{b(r)e^{2\tilde{A}(\bar{v},r)}-1}{r^2}}}. 
\end{equation}
\noindent The apparent horizon curve describes the time taken by a shell labelled by comoving radius $r$ to cross the apparent horizon. Hence, a sufficient condition for formation of a black hole as the end state of spherically symmetric gravitational collapse is that,
\begin{equation}\label{41}
 \forall\ r\in \mathbf{X}_2: t_{\text{ah}} (r) \leq t_s (r).
\end{equation}
\noindent On the other hand, if $t_{\text{ah}}(r) > t_s (r)$, then there may exist outgoing causal curves terminating in the past at the singularity, that may escape to future infinity. This corresponds to a naked singular end state. In light of the cosmic censorship conjectures of general relativity, these are of immediate interest, and we shall now turn towards an investigation of the formation of naked singularities as end states of gravitational collapse.

\subsection{Conditions for the Formation of Naked Singularities}

\noindent A naked singularity is (in principle) characterized by the existence of at least one past incomplete causal geodesic, i.e., we define a naked singularity by the existence of such a geodesic. However, for an observer, a single (null) geodesic emanating from the singularity corresponds to a single wavefront reaching them. Hence, to have observational and physical relevance, we investigate the conditions for the existence of a family of outgoing radial null geodesics (ORNGs) that constitute a set of non-zero measure. Moreover, it can be shown \cite{joshi2007gravitational} using the causal boundary constructions of Geroch, Kronheimer and Penrose \cite{21afc2b9-b08c-31a5-a52e-6412f7dd10e6}, that the existence of ORNGs emanating from a naked singularity implies the existence of non-radial null, as well as timelike geodesics emanating from the naked singularity. Hence, it is  sufficient to investigate only ORNGs, i.e., the existence of at least one past-incomplete ORNG is sufficient to say that the spacetime is naked singular. Firstly, we note that since the time curve is at least $\mathcal{C}^2$, it may be expanded around $r=0$, using Taylor's theorem as follows:
\begin{equation}\label{42}
    t(v,r) = t(v,0) + \sum_{k=1}^{j\geq 1} r^k \chi_k (v) + \mathscr{O}(r^{j+1}); \ \ \ \noindent \chi_k (v) \coloneqq \frac{1}{k!}\frac{d^k t}{dv^k}\bigg|_{r=0}.
\end{equation}
\noindent For ORNGs, the relation between the tangents is given by,
\begin{equation}\label{43}
   \frac{K^t}{K^r} \coloneqq \frac{dt/d\lambda}{dr/d\lambda} = e^{\psi(t,r)-\phi(t,r)} >0.
\end{equation}
\noindent Here $\lambda \in [\lambda_0,0) \subset \mathbb{R}$ is the affine parameter. We parameterize the geodesics such that $(v,r)\to (0,0)$ implies $\lambda\to 0$, i.e., the central singularity is identified by the vanishing of the affine parameter. In the limit of approach to the singularity, we must have $t\to t_s (r)$ and $R(t,r)\to 0$. We now introduce a new variable $u \coloneqq r^\alpha $, where $\alpha \geq 1$. Then the ORNG equation may be rewritten as, 
\begin{equation}\label{44}
    \frac{dR}{du} (t,r) = \frac{1}{\alpha} \frac{R'(t,r)}{r^{\alpha-1}}\left(1+\frac{\dot{R}(t,r)}{R'(t,r)} \frac{dt}{dr}\right).
\end{equation}
\noindent It is easy to see that provided $\dot{R}(t,r)<0$ and $R'(t,r)>0$ (which they are by construction), we have,
\begin{equation}\label{45}
 \forall (t,r) \in \mathbf{X}:  \frac{dR}{du} (t,r) >0  \implies \frac{dt}{dr} >0 .
\end{equation}
\noindent 
Defining $X(t,r) \coloneqq R(t,r)/ u$, the tangents to ORNGs are described by $dR/du$ in the $(R,u)$ plane. 
\textcolor{black}{Note that one could interpret $\alpha$ as follows: The physical radius $R(t,r)$ is of $\alpha$'th order in $r$ in the limit $t \to t_s(r)$ and $r\to 0$, i.e.,
\begin{equation}
    \lim_{r\to 0}\left(\lim_{t\to t_s(r)} R(t,r)\right) \sim \mathscr{O}\left(r^{\alpha}\right).
\end{equation}}
Using equations (\ref{12}), (\ref{13}) and (\ref{18}) in equation (\ref{44}) yields,
\begin{equation}\label{46}
    \frac{dR}{du} = \frac{1}{\alpha}\left(X + \dfrac{v^{\frac{N-3}{2}}v' r^{\frac{N-3}{2} + 2 -\alpha\left(\frac{N-1}{2}\right)}}{X^{\frac{N-3}{2}}}\right) \left(\frac{1- (F/R^{N-3})}{\sqrt{G}(\sqrt{G}+\sqrt{H})} \right) \eqcolon \mathcal{U }(X,u).
\end{equation}
\noindent Using the Pfaffian differential equation $dv = \dot{v}dt +v'dr = 0$, along a constant $v$ hypersurface, the expression for the term $v^{(N-3)/2}v'$  featuring in equation (\ref{46}) is obtained as,
\begin{equation}\label{47}
\begin{aligned}
    (v^{\frac{N-3}{2}} v')(t,r) &= e^{\psi(t,r)} \sqrt{v^{N-3} (b_0 (r) e^{2\tilde{A}(v,r)}+\tilde{h}(v,r)) +\tilde{\mathcal{M}}(v,r)})\\
    \tilde{h}(v,r) &\coloneqq \frac{e^{2\tilde{A}(v,r)}-1}{r^2}.
\end{aligned}
\end{equation}
\noindent Note that from equation (\ref{46}), we find that,
\begin{equation}\label{48}
  \lim_{(v,r) \to (0,r)} \frac{F}{R^{N-3}} = \infty \implies \lim_{(v,r) \to (0,r)} \frac{dR}{du} = -\infty.
\end{equation}
\noindent If the naked singularly can be characterized by the existence of positive $dR/du$ (in the limit $v\to 0$) as a simple root of the ORNG equation in the limit of approach to the singularity, the following statement holds: \textit{If a singularity is non-central, it is not naked; since in the neighbourhood of the singularity, there do not exist any ORNGs. This is because the tangents to the geodesics are negative.} Owing to the results of theorem \ref{t3} (to be discussed soon), we limit the scope of our paper to the cases where there exist positive values of $dR/du$ that constitute simple roots of the ORNG equation in the limit of approach to the singularity. In other words, we restrict ourselves to central singularities.
Since $X(t,r) \coloneqq R(t,r)/ u$, the tangents to the ORNGs may be expressed near the central singularity, using L'Hospital's rule as,
\begin{equation}\label{49}
 X_0 \coloneqq \lim_{(v,r)\to (0,0)}\frac{R}{u} =\lim_{(v,r)\to (0,0)} \frac{dR}{du}.
\end{equation}
\noindent %The necessary and sufficient condition for the central singularity to be naked is that there exists an ORNG terminating in the past at the singularity with a positive tangent. 
In what follows, we give a criterion for the central singularity to be naked. This generalises the work done in \cite{PhysRevD.47.5357} for the case of spherically symmetric inhomogeneous dust collapse, and the work done in \cite{PhysRevD.76.084026} for the specific cases of spherically symmetric type-I matter collapse with $\chi_1\neq 0$.\\

\begin{definition}\label{d01}
    (Positive root condition)\\
    The future development $(M, \mathbf{g})$ resulting from the $N (>3)$ dimensional spherically symmetric gravitational collapse of a general type-I matter field $($from regular initial data subject to the DEC, NEC and the Einstein equations$)$ is said to satisfy the positive root condition (PRC) if the function $V: \mathbb{R} \supset I\to \mathbb{R}$ defined by,
     \begin{equation}\label{50}
         V(X) = X - \mathcal{U}(X,0),
     \end{equation}
     \noindent admits at least one real positive root $X_0 \in \mathbb{R}^+$. Here,
\begin{equation}\label{51}
\begin{aligned}
 \mathcal{U}(X,0) &\coloneqq \frac{1}{\alpha} \left(X+\frac{A}{X^{\frac{N-3}{2}}}\right) \left(1+\frac{B}{X^{\frac{N-3}{2}}} \right)~\textrm{where}\\ 
     \mathbb{R} \ni A &\coloneqq \sqrt{\tilde{\mathcal{M}}_{00}} \left(\lim_{r\to 0} \left(\sum_{k=1}^{j\geq 1}  k\ r^{(k +1) +\frac{N-3}{2}  -\alpha\left(\frac{N-1}{2}\right)} \chi_k (0) + \mathscr{O}\left(r^{(j+3) +\frac{N-3}{2}  -\alpha\left(\frac{N-1}{2}\right)}\right)\right)\right);\\
    \mathbb{R} \ni B &\coloneqq -\sqrt{\tilde{\mathcal{M}}_{00}}\lim_{r\to 0}r^{\frac{N-1-\alpha(N-3)}{2}}; \ \ \ \tilde{\mathcal{M}}_{00} \coloneqq \lim_{(v,r)\to (0,0)} \tilde{\mathcal{M}}(v,r).
\end{aligned}
\end{equation}
\end{definition}

\begin{remark}\label{rem01}
\noindent  In definition \ref{d01}, for $V$ to be well defined we must have,
\begin{equation}\label{52} 
       \alpha \in \left\{ \frac{2k+N-1}{N-1}: \mathbb{N}\ni k \leq \frac{N-1}{N-3}\right\} % ~\text{and,} \\
%&\forall ~k\leq \frac{N-1}{2}(\alpha-1) : \chi_k (0) \sim \mathscr{O}\left(r^{\frac{N-1}{2}(\alpha-1) - k}\right).
\end{equation}
\noindent If (\ref{52}) does not hold, either $A$ or $B$ blows up to $\pm \infty$, thereby resulting in an ill-definition of $V$.\\
\end{remark}

\begin{definition}\label{d02}
(Simple positive root condition)\\
    The future development $(M,\mathbf{g})$ is said to satisfy the simple positive root condition (SPRC) if the function $V$ in definition \ref{d01} admits a simple positive real root.\\
\end{definition}

\begin{theorem}\label{t2}
(Positive root theorem - PRT)\\
     If the future development $(M, \mathbf{g})$ resulting from the $N$ dimensional spherically symmetric gravitational collapse of a general type-I matter field $($from regular initial data subject to the DEC, NEC and the Einstein equations$)$ contains at least one past-incomplete ORNG, then $(M, \mathbf{g})$ satisfies the positive root condition.
\end{theorem}
\begin{proof}
    Assume that $(M, \mathbf{g})$ contains at least one past-incomplete ORNG. It will follow the ORNG equation given by equation (\ref{46}). Using equation (\ref{30}) and (\ref{42}) and the Pfaffian differential equation $dv=0$, in the limit of approach to the central singularity, we have,
\begin{equation}\label{53}
    \lim_{(v,r)\to (0,0)} v^{\frac{N-3}{2}} v' = \sqrt{\tilde{\mathcal{M}}_{00}} \left(\left(\sum_{k=1}^{j\geq 1} k\ r^{k-1} \chi_k (0)\right) + \mathscr{O}(r^{j+1})\right).
\end{equation}
\noindent Using equations (\ref{27})-(\ref{29}), we find that,
\begin{equation}\label{54}
    \lim_{(t,r)\to(t_{s_0},0)} G = 1. 
\end{equation}
\noindent Using equation (\ref{54}) in the limiting case of equation (\ref{13}) yields,
\begin{equation}\label{55}
        \lim_{(t,r)\to(t_{s_0},0)} \sqrt{\frac{F}{R^{N-3}}} = 1-\frac{\sqrt{\tilde{\mathcal{M}}_{00}}}{X_0^{\frac{N-3}{2}}} \lim_{r\to0}  r^{\frac{N-1-\alpha(N-3)}{2}} \eqcolon 1+\frac{B}{X_0^{\frac{N-3}{2}}}.
\end{equation}
\noindent Using equations (\ref{53})-(\ref{55}), the ORNG equation (\ref{46}) can be evaluated in the limit of approach to the central singularity, to yield the required equation,
\begin{equation}\label{56}
  \begin{aligned}
     X_0 &= \frac{1}{\alpha} \left(X_0+\frac{\sqrt{\tilde{\mathcal{M}}_{00}}}{X_0^{\frac{N-3}{2}}} \left(\lim_{r\to 0} \left(\sum_{k=1}^{j\geq 1}  k\ r^{(k +1) +\frac{N-3}{2}  -\alpha\left(\frac{N-1}{2}\right)} \chi_k (0) + \mathscr{O}\left(r^{(j+3) +\frac{N-3}{2}  -\alpha\left(\frac{N-1}{2}\right)}\right)\right)\right)\right)\\ & \ \ \ \ \ \times \left(1-\frac{\sqrt{\tilde{\mathcal{M}}_{00}}}{X_0^{\frac{N-3}{2}}} \lim_{r\to0}  r^{\frac{N-1-\alpha(N-3)}{2}}\right)\\
     & \eqcolon \frac{1}{\alpha} \left(X_0+\frac{A}{X_0^{\frac{N-3}{2}}}\right) \left(1+\frac{B}{X_0^{\frac{N-3}{2}}}\right) \eqcolon \mathcal{U}(X_0,0).
\end{aligned}
\end{equation}
\noindent Here, using (\ref{45}), since we have considered ORNGs, we find that $X_0 \in \mathbb{R}^{+}$ provided equation (\ref{52}) holds.\\
\end{proof}
\noindent It is important to note that there may exist a class of values of $X_0$ that satisfy $V(X)=0$, however, it need not be true that each of these values of $X_0$ correspond to a tangent to singular null geodesics terminating at the central singularity in the past. In order to prove that the existence of a positive real root of $V(X)$ is sufficient to declare the central singularity naked, we must therefore assume first that there exists an $X=X_0 \in \mathbb{R}^+$ that satisfies $V(X)=0$, and then check if we can construct outgoing singular null geodesics with $X_0$ as its tangent in $(R,u)$ coordinates. We show this for the special case where $X_0$ is a simple root of $V$.\\

\begin{theorem}\label{t3}
(Simple positive root theorem - SPRT)\\
    If $(M, \mathbf{g})$ satisfies the simple positive root condition, then it contains at least one past-incomplete ORNG.
\end{theorem}
\begin{proof}
   \noindent Assume there exists a simple positive real root $X_0 \in \mathbb{R}^+$ of the function $V$ defined in theorem \ref{t2}. Then we may write \cite{PhysRevD.47.5357},
   \begin{equation}\label{57}
    V(X)=(X-X_0)(h_0 -1) + h(X).
\end{equation}
\noindent Here, $h_0 \coloneqq h(X_0)$ and $h(X)$ is to be chosen such that,
\begin{equation}\label{58}
    \frac{dh}{dX}\bigg|_{X=X_0} = 0.
\end{equation}
\noindent In other words, $h(X)$ must be chosen such that it contains only nonlinear terms in $(X-X_0)$. Differentiating equation (\ref{57}) yields,
\begin{equation}\label{59}
    h_0 = \frac{dV}{dX}\bigg|_{X=X_0} +1.
\end{equation}
\noindent Differentiating equation (\ref{50}) using equation (\ref{51}) and evaluating at $X=X_0$. Then substituting the result in equation (\ref{59}), we get the following expression for $h_0$:
\begin{equation}\label{60}
    h_0 = \frac{A (N-3) (2 B X_0^2 + X_0^{\frac{N + 1}{2}}) +B (N-5) X_0^{\frac{N + 3}{2}} + 2 X_0^N (\alpha -1)}{2 \alpha X_0^N}.
\end{equation}
\noindent Now consider the following equation in the ($X,u$) plane:
\begin{equation}\label{61}
    \frac{dX}{du} = \frac{1}{u}\left(\frac{dR}{du} - X\right) = \frac{\mathcal{U}(X,u)-X}{u}.
\end{equation}
\noindent The solutions of this equation (if they exist) are trajectories of RNGs in the $(X,u)$ plane of the form $u=u(X)$. If these trajectories terminate (in the past) at a singularity with tangents $X=X_0$, then as $u\to0$, we have $X\to  X_0$. Substituting equations (\ref{50}) and (\ref{60}) in equation (\ref{61}), and defining $S(X,u) \coloneqq h(X)+\mathcal{U}(X,u) - \mathcal{U}(X,0)$, we find that,
\begin{equation}\label{62}
    \frac{dX}{du} = \frac{(X-X_0)(h_0 -1)}{u}+\frac{S(X,u)}{u}.
\end{equation}
\noindent Multiplying this equation by $u^{h_0 -1}$ and integrating with respect to $u$ yields,
\begin{equation}\label{63}
X-X_0 = E u^{h_0 -1} + u^{h_0 -1} \int S(X,u) u^{-h_0}\ du.
\end{equation}
\noindent These represent a family of solutions to equation (\ref{61}). Since $X_0 \in \mathbb{R}^{+}$, we find that in the limit of approach to the central singularity, $dR/du >0$, and hence equation (\ref{63}) represents trajectories of past-incomplete ORNGs in the ($X,u$) plane. Note that $E$ is an integration constant that labels different ORNGs in the $(X,u)$ plane. In the limit of approach to the central singularity, i.e., in the limit $(X,u) \to (X_0 ,0)$, the second term in equation (\ref{63}) vanishes as $S\to 0$ by construction. However, the first term vanishes if either (i) $E=0$ and $h_0 \leq 1$, or (ii) $ E = 0$ and $h_0 >1$, or (iii) $E\neq 0$ and $h_0 >1$. In cases (i) and (ii), we find that a single ORNG terminates in the past with tangent $X=X_0$, and in case (iii), we find that an infinite family of ORNGs terminate at the singularity in the past with tangents $X=X_0$, each ORNG being labelled by a distinct value of the integration constant $E$. Using equation (\ref{60}), we find that if $V(X)$ is an increasing function of $X$ at $X=X_0$, then a non-zero measure set of ORNGs emanate from the singularity. In any case, if $X_0 \in \mathbb{R}^{+}$ satisfies equation (\ref{50}), then there exists at least one past-incomplete ORNG in $(M, \mathbf{g})$ that terminates in the past at the central singularity with tangent $X=X_0$.\\
\end{proof}

\noindent Theorems \ref{t2} and \ref{t3} provide necessary and sufficient conditions (respectively) for the central singularity to be naked. Moreover, the restriction on $\alpha$ given by equation (\ref{52}) also represents a necessary condition for the central singularity to be naked, since it is required for the well-definition of $V$. Note that substituting $k=1$ in equation (\ref{52}) yields the choice $\alpha = (N+1)/(N-1)$, which is the only value considered in \cite{PhysRevD.76.084026}. Moreover, substituting $N=4$ in equation (\ref{52}) recovers the results of \cite{PhysRevD.101.044052}, however with the correction that $k\leq 3$. Hence, we have generalized the results of \cite{PhysRevD.76.084026} for any value of the $\alpha$ parameter. 

\section{Gravitational Strength of the Naked Singularity}\label{sec4}

\noindent We have now established that there exist classes of solutions to the Einstein field equations subject to the DEC and NEC, that admit a naked singularity. It is often argued that singularities must only be taken seriously if they are of the so-called strong curvature type. The idea of a strong curvature singularity goes back to 1977, when Ellis and Schmidt \cite{Ellis:1977pj} first defined the gravitational strength of a singularity by referring to the destruction or crushing of objects approaching the singularity. In the same year, Tipler \cite{Tipler:1977zza} provided a mathematically rigorous formulation of this idea, by utilising the notion of the volume form defined by the Jacobi fields along causal geodesics approaching the singularity.\\

\color{black}
\begin{definition}
    (Non-trivial Jacobi fields)\\
    Let $(M,\mathbf{g})$ be an $N$ dimensional Lorentzian manifold and $\gamma: I \to M$ be a differentiable causal geodesic with tangents $\dot{\gamma} \in \Gamma(\mathcal{T}M)$. For each parameter value $\lambda \in I$, we define the subspace $H_{\gamma(\lambda) }M \subset\mathcal{T}_{\gamma(\lambda)} M$ as
    \begin{equation}
        H_{\gamma(\lambda) }M \coloneqq \{X \in \mathcal{T}_{\gamma(\lambda)} M\ |\ g_{\gamma(\lambda)}(X,\dot{\gamma}(\lambda)) = 0 \}.
    \end{equation}
    \noindent For null geodesics, we further define the quotient set $S_{\gamma(\lambda)}M \coloneqq H_{\gamma(\lambda)}M / \sim$ with respect to the equivalence relation for any $X,Y \in H_{\gamma(\lambda)}M$ given by $X \sim Y :\iff X-Y = c~\dot{\gamma} (\lambda)$, for some $c \in \mathbb{R}$. Then, in the case where $\gamma$ is timelike, the set of nontrivial Jacobi fields along $\gamma$ is defined as the $(N-1)$ dimensional subspace,
    \begin{equation}
        J^t_{\gamma(\lambda)}M \coloneqq \{J(\lambda) \in H_{\gamma(\lambda)} M\ |\ J\ \text{is a Jacobi field along}\ \gamma \}.
    \end{equation}
    \noindent Whereas in the case where $\gamma$ is null, this set is the $(N-2)$ dimensional subspace,
    \begin{equation}
        J^n_{\gamma(\lambda)} M \coloneqq \{J(\lambda) \in S_{\gamma(\lambda)} M\ |\ J\ \text{is a Jacobi field along}\ \gamma \}
    \end{equation}
\end{definition}

\begin{definition}
    (Volume form and volume along causal geodesics)\\
    The volume form for a timelike geodesic $\gamma: I \to M$, is defined by the wedge product of independent non-trivial Jacobi fields along $\gamma$ as,
\begin{equation}
\begin{aligned}
    \mu :  \bigtimes^{N-1} J^t_{\gamma(\lambda)}M &\to \bigwedge^{N-1} J^t_{\gamma(\lambda)} M\\
           (J_1 (\lambda) ,\cdots, J_{N-1} (\lambda)) &\mapsto J_1 (\lambda) \wedge \cdots \wedge J_{N-1} (\lambda)
\end{aligned}
\end{equation}
\noindent It is defined likewise for a null geodesic by replacing $N-1$ by $N-2$ and $J^t_{\gamma(\lambda)}M$ by $J^n_{\gamma(\lambda)}M$. In a basis $(e_1, \cdots e_{N-1})$ of $J^t_{\gamma(\lambda)} M$ (resp. in a basis $(e_1, \cdots, e_{N-m})$ of $J^n_{\gamma(\lambda)} M$), the non-trivial Jacobi fields will have components $J_1^a (\lambda),\cdots, J_{N-1}^a (\lambda)$ (resp. $J_1^a (\lambda),\cdots,J_{N-2}^a (\lambda)$). If $(\epsilon^1, \cdots, \epsilon^{N-1})$ (resp. $\epsilon^1,\cdots,\epsilon^{N-2}$) is the corresponding basis for $(J^t_{\gamma(\lambda)}M)^*$ (resp. $(J^n_{\gamma(\lambda)}M)^*$), then the volume corresponding to the volume element is the scalar given by,
\begin{equation}
    V(\lambda) \coloneqq \mu(\epsilon^1,\cdots,\epsilon^{N-m}) \equiv \det\ [J^a_b(\lambda)]
\end{equation}

\noindent where $m=1$ for timelike and $m=2$ for null geodesics, and $[J^a_b (\lambda)]$ is the matrix whose columns are defined by the components of the $(N-m)$ non-trivial Jacobi fields in the given basis, i.e., $a$ and $b$ run from $1$ to $N-m$. Since the determinant is basis independent, the volume thus defined is also basis independent. \\
\end{definition} 

\begin{definition}\label{d16}
    (Strong curvature singularity, Tipler \cite{Tipler:1977zza})\\
    Let $(M,\mathbf{g})$ be singular, and let $\gamma: [\lambda_0, 0) \to M$ be an incomplete and inextendible causal geodesic, such that its affine parameter $\lambda \to 0$ in the limit of approach to the singularity. Then the curve $\gamma$ is said to terminate strongly if,
\begin{equation}\label{65}
    \lim_{\lambda \to 0} V(\lambda) = 0.
\end{equation}
\noindent The singularity itself is called a strong curvature singularity if all causal geodesics of $M$ terminate strongly.\\
\end{definition}

\noindent In 1985, Clarke and Krolak \cite{1985JGP.....2..127C} gave a sufficient condition for a curve to terminate strongly, in terms of the rate of growth of the Ricci curvature in the limit of approach to the singularity.\\

\begin{definition}\label{d}
    (Strong limiting focussing condition, Clarke and Krolak \cite{1985JGP.....2..127C})\\
    Let $(M,\mathbf{g})$ be singular and satisfy the causal convergence condition, i.e., if for all causal curves $\gamma: \mathbb{R} \supseteq I \to M$ with parameter $\lambda\in I$ and tangents $\dot{\gamma}$,  we have,
    \begin{equation}
        \mathbf{Ric}(\dot{\gamma},\dot{\gamma}) \geq 0,
    \end{equation}

    \noindent then $\gamma$ is said to satisfy the strong limiting focussing condition (SLFC) if the following integral is non-integrable on $I$:
    \begin{equation}
        \int_{I} \mathbf{Ric}(\dot{\gamma},\dot{\gamma})\ d\lambda.
    \end{equation}
\end{definition}

\begin{proposition}\label{p1}
    (Curvature growth condition, Clarke and Krolak \cite{1985JGP.....2..127C})\\
    A causal geodesic $\gamma: [\lambda_0,0) \to M$ with tangents $\dot{\gamma}$ and $\lambda \to 0$ at the singularity, terminates strongly if the SLFC holds along $\gamma$. A sufficient condition for this is as follows:
\begin{equation}\label{66}
\lim_{\lambda\to 0} \lambda^2\ \mathbf{Ric}(\dot{\gamma},\dot{\gamma}) > 0.
\end{equation}
\noindent We refer to this as the curvature growth condition (CGC), since it implies a strong rate of growth of the Ricci curvature (in particular, a quadratic blow-up) along $\gamma$ as it approaches the singularity.\\
\end{proposition}
\begin{proof}
    The result can be obtained by a procedure somewhat similar to the well-known focusing theorem \cite{Wald:1984rg}. Consider the Raychaudhuri equation for a hypersurface-orthogonal congruence of causal geodesics,
    \begin{equation}
        \frac{d \Theta}{d\lambda} = - \mathbf{Ric}(\dot{\gamma},\dot{\gamma}) - \sigma^2 - \frac{1}{N-m} \Theta^2
    \end{equation}
    \noindent Here, $\Theta$ is the expansion and $\sigma$ is the shear of the congruence. Further, $m=1$ for timelike and $m=2$ for null congruences. Assume that the CGC holds along all curves in the congruence. Then we can ignore the shear and expansion terms on the RHS, and we get the inequality,
    \begin{equation}
        \frac{d\Theta}{d\lambda} \leq \frac{A}{\lambda^2}
    \end{equation}
\noindent for some $A>0$. Integrating over [$\Theta_0, \Theta$] yields,
\begin{equation}
    \Theta \leq C + \frac{A}{\lambda}
\end{equation}
 \noindent where $C$ is some constant due to integration. The expansion $\Theta$ measures the divergence or convergence of geodesics within a congruence, and as such is related to the volume spanned by Jacobi fields along these geodesics as follows \cite{hawking2023large}:
 \begin{equation}
     \Theta(\lambda) = \frac{1}{V(\lambda)}\frac{d}{d\lambda} V(\lambda)
 \end{equation}
\noindent In other words, the expansion measures the fractional rate of change of the volume enclosed by a geodesic congruence. Substituting this relation in the inequality, and integrating over [$V_0,V$] yields,
\begin{equation}
    V(\lambda) \leq B \lambda^A
\end{equation}
\noindent where $B$ is an integration constant. Clearly in the limit $\lambda \to 0$, we have $V(\lambda) \to 0$, since $A>0$. Hence the required result follows. 
\end{proof}
\color{black}

\noindent It follows that if the CGC is satisfied along all causal geodesics in $M$, then the singularity is by definition strong-curvature type. Moreover, for the purposes of this paper, the assumed spherical symmetry leads to the following useful result, the clear proof of which, the authors could not find in existing literature.\\
\begin{proposition}\label{p2}
    A spherically symmetric spacetime $(M,\mathbf{g})$ satisfying the DEC and NEC, with $\mathbf{g}$ given by equation $($\ref{1}$)$, contains a strong curvature singularity if the CGC holds along at least one ORNG $\gamma:[\lambda_0, 0) \to M$.
\end{proposition}
\begin{proof}
    We work in the comoving coordinates as defined earlier. In these coordinates, the components of the tangents to outgoing causal geodesics are related as follows \cite{PhysRevD.101.044052}:
\begin{subequations}
    \begin{align}
        \left(\frac{dt}{d\lambda}\right)^2 &= \frac{1}{e^{2\phi(t,r)}} \left[\left(\frac{dr}{d\lambda}\right)^2 e^{2\psi(t,r)} +\frac{\ell^2}{R^2 (t,r)} - \mathcal{B}(\lambda)\right]; \label{67a}  \\
        \left(\frac{d \theta^1}{d\lambda}\right)^2 &+ \sum_{i=2}^{N-2} \left(\prod_{j=1}^{i-1} \sin^2 \theta^j\right) \left(\frac{d\theta^i}{d\lambda}\right)^2 = \frac{\ell^2}{R^4 (t,r)}. \label{67b}
    \end{align}
\end{subequations}
\noindent Here, $\ell$ is a conserved quantity called the \textit{impact parameter} \cite{PhysRevD.101.044052}, which vanishes for any radial curve. Moreover, we have defined,
\begin{equation}\label{68}
 \forall\ \lambda \in [\lambda_0,0): ~\mathcal{B}(\lambda) \coloneqq \mathbf{g}_{\gamma(\lambda)}(\dot{\gamma}(\lambda),\dot{\gamma}(\lambda)).
\end{equation}
\noindent For outgoing causal geodesics, we have,
\begin{equation}\label{69}
  \mathbf{Ric} (\dot{\gamma}, \dot{\gamma}) = \mathbf{T}(\dot{\gamma},\dot{\gamma}) - \frac{\mathcal{B}(\lambda)}{N-2} \mathsf{tr}_\mathbf{g} \mathbf{T}.
\end{equation}
\noindent Here, $\mathbf{T}$ is the matter field over $M$. Using equations (\ref{2}), (\ref{67a}) and (\ref{67b}), and substituting $B(\lambda) = -1$, we find that for any non-radial timelike geodesic,
\begin{equation}
\begin{aligned}\label{70}
  \lim_{\lambda \to 0} \lambda^2\ \mathbf{Ric} (\dot{\gamma},\dot{\gamma}) =& \lim_{\lambda\to 0} \lambda^2 (\rho + p_r) e^{2\psi} \left(\frac{dr}{d\lambda}\right)^2\\
     &+ \lim_{\lambda \to 0} \lambda^2 (\rho + p_\theta ) \frac{\ell^2}{R^2 } \\
     &+\lim_{\lambda \to 0} \frac{\lambda^2}{N-2} (\rho+p_r)\\
     &+ \lim_{\lambda \to 0} \lambda^2 (\rho+p_\theta).
     \end{aligned}
\end{equation}
\noindent For radial timelike geodesics, the second term in equation (\ref{70}) vanishes (taking $\ell=0$), and for non-radial null geodesics, the last two terms vanish (taking $B(\lambda)=0$). Note that the first term on the RHS turns out to be common among radial null geodesics and all other kinds of causal geodesics. Moreover, from equation (\ref{8}), it is clear that the NEC \cite{Maeda:2018hqu} implies that the all the terms on the RHS of equation (\ref{70}) are non-negative. For the first term on the RHS of equation (\ref{68}), since $r\to 0 $ as $\lambda \to 0$, using L'Hospital's rule, we find that,
\begin{equation}\label{71}
    \lim_{\lambda \to 0}  \lambda^2 (\rho+ p_r) e^{2\psi} \left(\frac{dr}{d\lambda}\right)^2 = \lim_{\lambda\to 0 } (\rho + p_r) r^2 e^{2\psi}.
\end{equation}
\noindent This limit is the only term featuring in the expression for the CGC for ORNGs. Since all the other terms in equation (\ref{70}) are non-negative, we have the useful result that if the CGC holds along an ORNG, then it holds along all causal curves, and hence, using proposition \ref{p1} and definitions \ref{d16} and \ref{d}, we find that the singularity is of strong-curvature type.\\
\end{proof}

\noindent As before, it is useful to characterize the properties of the singularity in terms of the parameter $\alpha$. The positivity (or otherwise) of the limit in equation (\ref{71}) depends on what values $\alpha$ takes. In other words, the requirement that the central naked singularity be of strong-curvature type leads to a more restriction on $\alpha$ that is already restricted as in (\ref{52}). \textcolor{black}{Using the Einstein equations (\ref{10a}) and (\ref{10b}), and equations (\ref{12}), (\ref{43}) and (\ref{71}), the CGC for any ORNG $\gamma: [\lambda_0,0)\to M$ may be rewritten by evaluating the limit in the $(v,r)$ plane as},
\begin{equation}\label{72}
    \lim_{\lambda \to 0} \lambda^2\ \mathbf{Ric} (\dot{\gamma},\dot{\gamma}) = \frac{N-2}{2}\lim_{(t,r) \to (t_{s_0},0)} \left(\frac{F'}{R'}-\frac{\dot{F}}{\dot{R}}\right) \frac{(R')^2}{R^{N-2}}.
\end{equation}
\noindent Using equation (\ref{41}) and (\ref{46}), the limit of $R'$ can be written as,
\begin{equation}\label{73}
    \lim_{\lambda \to 0} R' = \alpha X_0 \lim_{r\to 0} r^{\alpha -1}.
\end{equation}
\noindent Using this result, in addition to equation (\ref{24}) and the requirement that $\tilde{\mathcal{M}}$ must be least $\mathcal{C}^1$ in the limit of approach to the singular epoch, the CGC for ORNGs is satisfied if
\begin{equation}\label{74}
     \frac{\alpha (N-2) \tilde{\mathcal{M}}(0,0)}{2 X_0 ^{N-3}} \left(\lim_{r\to 0}  r^{(3-N)\alpha + N} \right) + \frac{\alpha(N-2)\tilde{\mathcal{M}}(0,0)}{2X_0^{N-3}} \left(\lim_{r\to 0} r^{(3-N)\alpha + (N-1)}\right) > 0. 
\end{equation}
\noindent This inequality restricts the values that $\alpha$ can take in order for the central singularity to be of the strong curvature type. Indeed, from equation (\ref{74}) we find that for any ORNG $\gamma:[\lambda_0,0)\to M$, we have,
\begin{equation}\label{75}
   \alpha \in \left[\frac{N-1}{N-3}, \infty \right) \iff \lim_{\lambda \to 0} \lambda^2\ \mathbf{Ric} (\dot{\gamma},\dot{\gamma}) >0. 
\end{equation}
\noindent In other words, the limit in (\ref{74}) diverges, and hence the ORNG satisfies the CGC. Hence, by propositions \ref{p1} and \ref{p2}, all causal geodesics terminate strongly, and hence the singularity is of the strong curvature type. \\
\begin{theorem}\label{t4}
Central naked singularities for $N\in\{4,5\}$ are of strong-curvature type if
\begin{equation}\label{76}
 \alpha = \frac{N-1}{N-3}.
\end{equation}
\end{theorem}
\begin{proof}
This result follows directly from equations (\ref{52}) and (\ref{75}), by employing theorems \ref{t2} and \ref{t3}, and propositions \ref{p1} and \ref{p2}. For the singularity to be naked, equation (\ref{51}) must hold. For the singularity to be strong, equation (\ref{75}) must hold. We break down our proof into two cases:
 \begin{enumerate}[(i)]
     \item For naked singularities in $N=4$, it is necessary to have,
     \begin{equation}\label{77}
         \alpha \in \left\{ \frac{5}{3},\frac{7}{3},3 \right\},
     \end{equation}
     \noindent and for strong curvature singularities in $N=4$, it is sufficient to have $\alpha =3$. Hence for $\alpha = (N-1)/(N-3)=3$ we get strong curvature naked singularities in 4 dimensions.

     \item For naked singularities in $N=5$, it is necessary to have,
     \begin{equation}\label{78}
         \alpha \in \left\{\frac{3}{2},2\right\},
     \end{equation}
     \noindent and for strong curvature singularities in $N=5$, it is sufficient to have $\alpha = 2$. Hence for $\alpha = (N-1)/(N-3)=2$ we get strong curvature naked singularities in 5 dimensions.
 \end{enumerate}
  \noindent In any case, we find that equation (\ref{76}) must hold.\\
\end{proof}
    
\begin{proposition}\label{t5}
    Let $(M,\mathbf{g})$ be the future development arising from the $N\geq 6$ dimensional collapse of type-I matter fields subject to the WEC, DEC and Einstein equations. Further, let $\gamma: [\lambda_0,0) \to M$ be a past-incomplete geodesic that terminates at a central naked singularity in the limit $\lambda \to 0$. Then $\gamma$ does not satisfy the curvature growth condition.
\end{proposition}
\begin{proof}
    It follows from equation (\ref{52}) that for naked singularities in $N\geq 6$, it is necessary to have,
     \begin{equation}\label{79}
         \alpha = \frac{N+1}{N-1},
     \end{equation}
     \noindent For strong curvature singularities in $N\geq 6$, it is sufficient to have $\alpha \geq (N-1)/(N-3)$ (using equation (\ref{75})). Combining these values of $\alpha$ with the value of $\alpha$ given by equation (\ref{79}), we find that in $N \geq 6$, we have strong curvature naked singularities provided,
     \begin{equation}\label{80}
         \frac{N+1}{N-1} \geq \frac{N-1}{N-3}. 
     \end{equation}
     \noindent This inequality leads to a contradiction and hence is not realized. Hence if there exist past-incomplete curves in the future development, then these curves do not satisfy the curvature growth condition of Clarke and Krolak.\\
\end{proof} 

    \begin{remark}\label{r0001}
        \noindent Note that equation (\ref{52}) does not apply in the case $N=3$; hence, in this case, we cannot proceed using the methodology of theorem \ref{t4}. \\
    \end{remark}
    
\noindent The cosmic censorship conjectures as discussed in section \ref{sec1} require that all genuine singularities (i.e., singularities through which the spacetime is inextendible) must be covered by a horizon. However, through our analysis here, we have shown that there exists a class of solutions to the Einstein field equations evolving from spherically symmetric regular initial data (in at least $N=4$ and $N=5$), that contain strong curvature naked singularities. \textcolor{black}{The spacetimes containing Tipler naked singularities will be $C^2$ inextendible in the sense that there does not exist any other spacetime with a $C^2$ metric, in which the original spacetime can be isometrically embedded. This happens due to the fact that the CGC implies a blow up of the Ricci curvature, and hence of the second order derivatives of the metric in the limit of approach to the central singularity.} In particular this means that there exist maximal \textit{non-globally hyperbolic} inextendible solutions to the Einstein equations in at least four and five dimensions. In other words, theorem \ref{t4} may be interpreted as a violation of some form of the cosmic censorship conjectures of general relativity, since these occur for a wide class of matter fields - namely type-I matter fields. On the other hand, we may interpret proposition \ref{t5} as being indicative of a \textit{possible} restoration of cosmic censorship in higher dimensions, in the sense that for the unhindered gravitational collapse of type-I matter fields in $N\geq 6$, if there exists a subclass of the initial data which evolves to a maximal development containing past incomplete ORNGs, then the naked singularity identified by these ORNGs could possibly be gravitationally weak. These results are remarkably similar to the ones obtained by Giambo and Quintavalle \cite{Giambo:2007ps}, where it is shown that ``phase transitions" between black holes and naked singularities that exist in $N=4$ and $N=5$, cease to exist beyond 5 dimensions.

\section{Concluding Remarks}\label{sec5}

\noindent Here we summarize the results obtained, and thereby provide some future directions and scope for further research.

\begin{enumerate}[(i)]
    \item In order to precisely formulate the notion of generic initial data for cosmic censorship, it is imperative to investigate the gravitational collapse of specific forms of matter that leads to the formation of a naked singularity. In this spirit, we constructed here the sufficient conditions for the existence of naked singularities as the end state of the spherically symmetric collapse of general type-I matter fields in an arbitrary and finite number of dimensions $N$. This analysis was previously conducted in \cite{PhysRevD.76.084026}, but with a fixed value of the parameter $\alpha = (N+1)/(N-1)$, or equivalently for the specific case $\chi_1 (v)\neq 0$. In the present paper, we drop this specialisation and later restrict the range of $\alpha$ by imposing the necessary and sufficient conditions for the formation of a naked singularity.

    \item The analysis of the strength of the central naked singularity forming in higher dimensions was not considered in \cite{PhysRevD.76.084026}. As such, the physical relevance of the singularities discussed therein is not accounted for. Here, we derived sufficient conditions for the naked singularity to be of the strong curvature type (Tipler strong \cite{Tipler:1977zza}) in terms of the parameter $\alpha$. This leads to a class of inextendible spacetimes that admit a strong curvature naked singularity. This investigation shows that the naked singularities discussed in \cite{PhysRevD.76.084026} need not necessarily be strong curvature singularities.

    \item An important point to consider here is the stability and genericity of the naked singularities in question. One such notion of stability is dynamical stability. To be considered significant, these naked singularities must not be removable by a perturbation of the initial data, and they should form for a sufficiently dense class of initial data.  We have not conducted any such analysis here. Following the general theme of dynamical relativity, one must again emphasize that it is the initial data that must be judged and not the final solution. Strictly speaking, one must check whether a generic set of choices for the initial data leads to a central naked singularity or not; only then can one speak about the genericity of the naked singularities. In other words, it is crucial to formalize the notion of genericity of the naked singularity, in order to make judgements about the cosmic censorship conjectures. Hence, a possible future scope would be to construct an appropriate notion of genericity for spherically symmetric gravitational collapse. These definitions may be motivated, for instance, from the notion of generic properties of vector fields in the study of dynamical systems \cite{abraham2008foundations}, or from topologically induced notion of genericity for some suitable function space of initial data \cite{joshi2007gravitational}.

    \item The naked singularities discussed herein may be either locally or globally naked, i.e., we have not investigated whether ORNGs escape the collapsing cloud before or after the apparent horizon does. Although the general conditions under which a singularity will be globally naked are available in the existing literature \cite{PhysRevD.47.5357}, one needs to further investigate under which conditions do these general conditions hold. 

    \item Although we did not discuss this explicitly for the general class of spacetimes considered here, it is known that for a few specific subclasses (like tangential pressure models and dust models), all the singularities for $N\geq 6$ are black hole singularities \cite{PhysRevD.69.104002,PhysRevD.72.024006,Banerjee:2002sy}. This is due to the fact that the Taylor expansion of the apparent horizon curve contains a term that becomes significant beyond a critical dimension, after which the apparent horizon curve has a decreasing behaviour. We find a somewhat similar result here, where for $N\geq 6$, a sufficient criterion for checking the curvature strength of the singularity fails to hold. Incidentally, Giambo and Quintavalle \cite{Giambo:2007ps} have obtained similar results regarding ``phase transitions" between naked singularities and black holes. In particular, they have shown that while such transitions are possible in $N=4$ and $N=5$ dimensions, they cease to exist for $N\geq 6$. This indicates that the existence of phase transitions between naked singualrities and black holes may be related to the rate of growth of curvature in the limit of approach to these singularities. The formation of strong curvature naked singularities could be physically better understood in the context of quantum gravitationally motivated regimes \cite{Chakraborty:2017uku,Alday:2019qrf,Lu:2008jk,Frassino:2022zaz,Pourhassan:2017kmm,Gomez-Fayren:2023wxk,Hendi:2021yii,Deo:2023vvb,PhysRevLett.96.031302,Emparan:2008eg,Maldacena:1997re,Gubser:2000nd}, for instance by analyzing the physical effects of extra spatial dimensions in such theories.
\end{enumerate}

\backmatter 

\bmhead*{Acknowledgments} KNS would like to thank Oem Trivedi for discussions regarding the interpretation of the results obtained in terms of quantum gravitationally motivated concepts. We would like to thank Roberto Giambo for valuable clarifications and discussions regarding similarity in dimensional dependence of CGC and phase transition properties of singularities. Additionally we also thank Jun-Qi Guo for bringing into notice the papers \cite{Dwivedi:1994qs}, \cite{PhysRevLett.60.241}, and \cite{PhysRevLett.59.2137}. 

\bmhead*{Data availability} Data sharing is not applicable to this article as no datasets were generated or analysed during the current study.

\bmhead*{Declarations}

\bmhead*{Conflict of interest} The authors have no competing interests to declare that are relevant to the content of this article.
\bibliography{main}% common bib file


\begin{thebibliography}{85}
\ifx \bisbn   \undefined \def \bisbn  #1{ISBN #1}\fi
\ifx \binits  \undefined \def \binits#1{#1}\fi
\ifx \bauthor  \undefined \def \bauthor#1{#1}\fi
\ifx \batitle  \undefined \def \batitle#1{#1}\fi
\ifx \bjtitle  \undefined \def \bjtitle#1{#1}\fi
\ifx \bvolume  \undefined \def \bvolume#1{\textbf{#1}}\fi
\ifx \byear  \undefined \def \byear#1{#1}\fi
\ifx \bissue  \undefined \def \bissue#1{#1}\fi
\ifx \bfpage  \undefined \def \bfpage#1{#1}\fi
\ifx \blpage  \undefined \def \blpage #1{#1}\fi
\ifx \burl  \undefined \def \burl#1{\textsf{#1}}\fi
\ifx \doiurl  \undefined \def \doiurl#1{\url{https://doi.org/#1}}\fi
\ifx \betal  \undefined \def \betal{\textit{et al.}}\fi
\ifx \binstitute  \undefined \def \binstitute#1{#1}\fi
\ifx \binstitutionaled  \undefined \def \binstitutionaled#1{#1}\fi
\ifx \bctitle  \undefined \def \bctitle#1{#1}\fi
\ifx \beditor  \undefined \def \beditor#1{#1}\fi
\ifx \bpublisher  \undefined \def \bpublisher#1{#1}\fi
\ifx \bbtitle  \undefined \def \bbtitle#1{#1}\fi
\ifx \bedition  \undefined \def \bedition#1{#1}\fi
\ifx \bseriesno  \undefined \def \bseriesno#1{#1}\fi
\ifx \blocation  \undefined \def \blocation#1{#1}\fi
\ifx \bsertitle  \undefined \def \bsertitle#1{#1}\fi
\ifx \bsnm \undefined \def \bsnm#1{#1}\fi
\ifx \bsuffix \undefined \def \bsuffix#1{#1}\fi
\ifx \bparticle \undefined \def \bparticle#1{#1}\fi
\ifx \barticle \undefined \def \barticle#1{#1}\fi
\bibcommenthead
\ifx \bconfdate \undefined \def \bconfdate #1{#1}\fi
\ifx \botherref \undefined \def \botherref #1{#1}\fi
\ifx \url \undefined \def \url#1{\textsf{#1}}\fi
\ifx \bchapter \undefined \def \bchapter#1{#1}\fi
\ifx \bbook \undefined \def \bbook#1{#1}\fi
\ifx \bcomment \undefined \def \bcomment#1{#1}\fi
\ifx \oauthor \undefined \def \oauthor#1{#1}\fi
\ifx \citeauthoryear \undefined \def \citeauthoryear#1{#1}\fi
\ifx \endbibitem  \undefined \def \endbibitem {}\fi
\ifx \bconflocation  \undefined \def \bconflocation#1{#1}\fi
\ifx \arxivurl  \undefined \def \arxivurl#1{\textsf{#1}}\fi
\csname PreBibitemsHook\endcsname

%%% 1
\bibitem[\protect\citeauthoryear{Ringstr\"om}{2015}]{Ringstrom:2015jza}
\begin{barticle}
\bauthor{\bsnm{Ringstr\"om}, \binits{H.}}:
\batitle{{Origins and development of the Cauchy problem in general relativity}}.
\bjtitle{Class. Quant. Grav.}
\bvolume{32}(\bissue{12}),
\bfpage{124003}
(\byear{2015})
\doiurl{10.1088/0264-9381/32/12/124003}
\end{barticle}
\endbibitem

%%% 2
\bibitem[\protect\citeauthoryear{{Four{\`e}s-Bruhat}}{1952}]{1952AcMa...88..141F}
\begin{barticle}
\bauthor{\bsnm{{Four{\`e}s-Bruhat}}, \binits{Y.}}:
\batitle{{Th{\'e}or{\`e}me d'existence pour certains syst{\`e}mes d'{\'e}quations aux d{\'e}riv{\'e}es partielles non lin{\'e}aires}}.
\bjtitle{Acta Mathematica}
\bvolume{88}(\bissue{1}),
\bfpage{141}--\blpage{225}
(\byear{1952})
\doiurl{10.1007/BF02392131}
\end{barticle}
\endbibitem

%%% 3
\bibitem[\protect\citeauthoryear{Choquet-Bruhat and Geroch}{1969}]{Choquet-Bruhat:1969ywq}
\begin{barticle}
\bauthor{\bsnm{Choquet-Bruhat}, \binits{Y.}},
\bauthor{\bsnm{Geroch}, \binits{R.P.}}:
\batitle{{Global aspects of the Cauchy problem in general relativity}}.
\bjtitle{Commun. Math. Phys.}
\bvolume{14},
\bfpage{329}--\blpage{335}
(\byear{1969})
\doiurl{10.1007/BF01645389}
\end{barticle}
\endbibitem

%%% 4
\bibitem[\protect\citeauthoryear{Sbierski}{2016}]{Sbierski:2013kca}
\begin{barticle}
\bauthor{\bsnm{Sbierski}, \binits{J.}}:
\batitle{{On the Existence of a Maximal Cauchy Development for the Einstein Equations: a Dezornification}}.
\bjtitle{Annales Henri Poincare}
\bvolume{17}(\bissue{2}),
\bfpage{301}--\blpage{329}
(\byear{2016})
\doiurl{10.1007/s00023-015-0401-5}
{\href{https://arxiv.org/abs/1309.7591}{{arXiv:1309.7591}}}
{[gr-qc]}
\end{barticle}
\endbibitem

%%% 5
\bibitem[\protect\citeauthoryear{Dafermos and Rodnianski}{2013}]{Dafermos:2008en}
\begin{barticle}
\bauthor{\bsnm{Dafermos}, \binits{M.}},
\bauthor{\bsnm{Rodnianski}, \binits{I.}}:
\batitle{{Lectures on black holes and linear waves}}.
\bjtitle{Clay Math. Proc.}
\bvolume{17},
\bfpage{97}--\blpage{205}
(\byear{2013})
{\href{https://arxiv.org/abs/0811.0354}{{arXiv:0811.0354}}}
{[gr-qc]}
\end{barticle}
\endbibitem

%%% 6
\bibitem[\protect\citeauthoryear{Hawking and Penrose}{1970}]{Hawking:1970zqf}
\begin{barticle}
\bauthor{\bsnm{Hawking}, \binits{S.W.}},
\bauthor{\bsnm{Penrose}, \binits{R.}}:
\batitle{{The Singularities of gravitational collapse and cosmology}}.
\bjtitle{Proc. Roy. Soc. Lond. A}
\bvolume{314},
\bfpage{529}--\blpage{548}
(\byear{1970})
\doiurl{10.1098/rspa.1970.0021}
\end{barticle}
\endbibitem

%%% 7
\bibitem[\protect\citeauthoryear{Senovilla}{1998}]{Senovilla:1998oua}
\begin{barticle}
\bauthor{\bsnm{Senovilla}, \binits{J.M.M.}}:
\batitle{{Singularity Theorems and Their Consequences}}.
\bjtitle{Gen. Rel. Grav.}
\bvolume{30},
\bfpage{701}
(\byear{1998})
\doiurl{10.1023/A:1018801101244}
{\href{https://arxiv.org/abs/1801.04912}{{arXiv:1801.04912}}}
{[gr-qc]}
\end{barticle}
\endbibitem

%%% 8
\bibitem[\protect\citeauthoryear{Steinbauer}{2022}]{Steinbauer_2022}
\begin{barticle}
\bauthor{\bsnm{Steinbauer}, \binits{R.}}:
\batitle{The singularity theorems of general relativity and their low regularity extensions}.
\bjtitle{Jahresbericht der Deutschen Mathematiker-Vereinigung}
\bvolume{125}(\bissue{2}),
\bfpage{73}--\blpage{119}
(\byear{2022})
\doiurl{10.1365/s13291-022-00263-7}
\end{barticle}
\endbibitem

%%% 9
\bibitem[\protect\citeauthoryear{{Penrose}}{1974}]{1974IAUS...63..263P}
\begin{bchapter}
\bauthor{\bsnm{{Penrose}}, \binits{R.}}:
\bctitle{{Singularities in cosmology.}}
In: \beditor{\bsnm{{Longair}}, \binits{M.S.}} (ed.)
\bbtitle{Confrontation of Cosmological Theories with Observational Data},
vol. \bseriesno{63},
pp. \bfpage{263}--\blpage{271}
(\byear{1974})
\end{bchapter}
\endbibitem

%%% 10
\bibitem[\protect\citeauthoryear{Penrose}{1969}]{Penrose:1969pc}
\begin{barticle}
\bauthor{\bsnm{Penrose}, \binits{R.}}:
\batitle{{Gravitational collapse: The role of general relativity}}.
\bjtitle{Riv. Nuovo Cim.}
\bvolume{1},
\bfpage{252}--\blpage{276}
(\byear{1969})
\doiurl{10.1023/A:1016578408204}
\end{barticle}
\endbibitem

%%% 11
\bibitem[\protect\citeauthoryear{Wald}{1984}]{Wald:1984rg}
\begin{bbook}
\bauthor{\bsnm{Wald}, \binits{R.M.}}:
\bbtitle{General Relativity}.
\bpublisher{Chicago Univ. Pr.},
\blocation{Chicago, USA}
(\byear{1984}).
\doiurl{10.7208/chicago/9780226870373.001.0001}
\end{bbook}
\endbibitem

%%% 12
\bibitem[\protect\citeauthoryear{Chru{\'s}ciel}{2020}]{chrusciel2020geometry}
\begin{bbook}
\bauthor{\bsnm{Chru{\'s}ciel}, \binits{P.T.}}:
\bbtitle{Geometry of Black Holes}
vol. \bseriesno{169}.
\bpublisher{Oxford University Press},
\blocation{Oxford}
(\byear{2020})
\end{bbook}
\endbibitem

%%% 13
\bibitem[\protect\citeauthoryear{Kommemi}{2013}]{Kommemi:2011wh}
\begin{barticle}
\bauthor{\bsnm{Kommemi}, \binits{J.}}:
\batitle{{The Global structure of spherically symmetric charged scalar field spacetimes}}.
\bjtitle{Commun. Math. Phys.}
\bvolume{323},
\bfpage{35}--\blpage{106}
(\byear{2013})
\doiurl{10.1007/s00220-013-1759-1}
{\href{https://arxiv.org/abs/1107.0949}{{arXiv:1107.0949}}}
{[gr-qc]}
\end{barticle}
\endbibitem

%%% 14
\bibitem[\protect\citeauthoryear{Choquet-Bruhat}{2009}]{choquet2009general}
\begin{bbook}
\bauthor{\bsnm{Choquet-Bruhat}, \binits{Y.}}:
\bbtitle{General Relativity and the Einstein Equations}.
\bsertitle{Oxford Mathematical Monographs}.
\bpublisher{OUP Oxford},
\blocation{Oxford}
(\byear{2009})
\end{bbook}
\endbibitem

%%% 15
\bibitem[\protect\citeauthoryear{Christodoulou}{1999}]{Demetrios}
\begin{barticle}
\bauthor{\bsnm{Christodoulou}, \binits{D.}}:
\batitle{On the global initial value problem and the issue of singularities}.
\bjtitle{Classical and Quantum Gravity}
\bvolume{16}(\bissue{12A}),
\bfpage{23}
(\byear{1999})
\doiurl{10.1088/0264-9381/16/12A/302}
\end{barticle}
\endbibitem

%%% 16
\bibitem[\protect\citeauthoryear{Joshi}{2014}]{Joshi:2013xoa}
\begin{bbook}
\bauthor{\bsnm{Joshi}, \binits{P.S.}}:
In: \beditor{\bsnm{Ashtekar}, \binits{A.}},
\beditor{\bsnm{Petkov}, \binits{V.}} (eds.)
\bbtitle{{Spacetime Singularities}},
pp. \bfpage{409}--\blpage{436}
(\byear{2014}).
\doiurl{10.1007/978-3-642-41992-8_20}
\end{bbook}
\endbibitem

%%% 17
\bibitem[\protect\citeauthoryear{Hawking}{1976}]{PhysRevD.14.2460}
\begin{barticle}
\bauthor{\bsnm{Hawking}, \binits{S.W.}}:
\batitle{Breakdown of predictability in gravitational collapse}.
\bjtitle{Phys. Rev. D}
\bvolume{14},
\bfpage{2460}--\blpage{2473}
(\byear{1976})
\doiurl{10.1103/PhysRevD.14.2460}
\end{barticle}
\endbibitem

%%% 18
\bibitem[\protect\citeauthoryear{Eardley and Smarr}{1979}]{PhysRevD.19.2239}
\begin{barticle}
\bauthor{\bsnm{Eardley}, \binits{D.M.}},
\bauthor{\bsnm{Smarr}, \binits{L.}}:
\batitle{Time functions in numerical relativity: Marginally bound dust collapse}.
\bjtitle{Phys. Rev. D}
\bvolume{19},
\bfpage{2239}--\blpage{2259}
(\byear{1979})
\doiurl{10.1103/PhysRevD.19.2239}
\end{barticle}
\endbibitem

%%% 19
\bibitem[\protect\citeauthoryear{Joshi et~al.}{2024}]{PhysRevD.109.064019}
\begin{barticle}
\bauthor{\bsnm{Joshi}, \binits{A.B.}},
\bauthor{\bsnm{Mosani}, \binits{K.}},
\bauthor{\bsnm{Joshi}, \binits{P.S.}}:
\batitle{Future null singularity due to gravitational collapse}.
\bjtitle{Phys. Rev. D}
\bvolume{109},
\bfpage{064019}
(\byear{2024})
\doiurl{10.1103/PhysRevD.109.064019}
\end{barticle}
\endbibitem

%%% 20
\bibitem[\protect\citeauthoryear{Joshi}{2007}]{joshi2007gravitational}
\begin{bbook}
\bauthor{\bsnm{Joshi}, \binits{P.S.}}:
\bbtitle{Gravitational Collapse and Spacetime Singularities}
vol. \bseriesno{2}.
\bpublisher{Cambridge University Press Cambridge},
\blocation{Cambridge}
(\byear{2007})
\end{bbook}
\endbibitem

%%% 21
\bibitem[\protect\citeauthoryear{Oppenheimer and Snyder}{1939}]{PhysRev.56.455}
\begin{barticle}
\bauthor{\bsnm{Oppenheimer}, \binits{J.R.}},
\bauthor{\bsnm{Snyder}, \binits{H.}}:
\batitle{On continued gravitational contraction}.
\bjtitle{Phys. Rev.}
\bvolume{56},
\bfpage{455}--\blpage{459}
(\byear{1939})
\doiurl{10.1103/PhysRev.56.455}
\end{barticle}
\endbibitem

%%% 22
\bibitem[\protect\citeauthoryear{Christodoulou}{1984}]{Christodoulou:1984mz}
\begin{barticle}
\bauthor{\bsnm{Christodoulou}, \binits{D.}}:
\batitle{{Violation of cosmic censorship in the gravitational collapse of a dust cloud}}.
\bjtitle{Commun. Math. Phys.}
\bvolume{93},
\bfpage{171}--\blpage{195}
(\byear{1984})
\doiurl{10.1007/BF01223743}
\end{barticle}
\endbibitem

%%% 23
\bibitem[\protect\citeauthoryear{Newman}{1986}]{Newman:1985gt}
\begin{barticle}
\bauthor{\bsnm{Newman}, \binits{R.P.A.C.}}:
\batitle{{Strengths of naked singularities in Tolman-Bondi space-times}}.
\bjtitle{Class. Quant. Grav.}
\bvolume{3},
\bfpage{527}--\blpage{539}
(\byear{1986})
\doiurl{10.1088/0264-9381/3/4/007}
\end{barticle}
\endbibitem

%%% 24
\bibitem[\protect\citeauthoryear{{Clarke} and {Kr{\'o}lak}}{1985}]{1985JGP.....2..127C}
\begin{barticle}
\bauthor{\bsnm{{Clarke}}, \binits{C.J.S.}},
\bauthor{\bsnm{{Kr{\'o}lak}}, \binits{A.}}:
\batitle{{Conditions for the occurence of strong curvature singularities}}.
\bjtitle{Journal of Geometry and Physics}
\bvolume{2}(\bissue{2}),
\bfpage{127}--\blpage{143}
(\byear{1985})
\doiurl{10.1016/0393-0440(85)90012-9}
\end{barticle}
\endbibitem

%%% 25
\bibitem[\protect\citeauthoryear{Tipler}{1977}]{Tipler:1977zza}
\begin{barticle}
\bauthor{\bsnm{Tipler}, \binits{F.J.}}:
\batitle{{Singularities in conformally flat spacetimes}}.
\bjtitle{Phys. Lett. A}
\bvolume{64},
\bfpage{8}--\blpage{10}
(\byear{1977})
\doiurl{10.1016/0375-9601(77)90508-4}
\end{barticle}
\endbibitem

%%% 26
\bibitem[\protect\citeauthoryear{Joshi and Dwivedi}{1993}]{PhysRevD.47.5357}
\begin{barticle}
\bauthor{\bsnm{Joshi}, \binits{P.S.}},
\bauthor{\bsnm{Dwivedi}, \binits{I.H.}}:
\batitle{Naked singularities in spherically symmetric inhomogeneous tolman-bondi dust cloud collapse}.
\bjtitle{Phys. Rev. D}
\bvolume{47},
\bfpage{5357}--\blpage{5369}
(\byear{1993})
\doiurl{10.1103/PhysRevD.47.5357}
\end{barticle}
\endbibitem

%%% 27
\bibitem[\protect\citeauthoryear{Dwivedi and Joshi}{1994}]{Dwivedi:1994qs}
\begin{barticle}
\bauthor{\bsnm{Dwivedi}, \binits{I.H.}},
\bauthor{\bsnm{Joshi}, \binits{P.S.}}:
\batitle{{On the occurrence of naked singularity in spherically symmetric gravitational collapse}}.
\bjtitle{Commun. Math. Phys.}
\bvolume{166},
\bfpage{117}--\blpage{128}
(\byear{1994})
\doiurl{10.1007/BF02099303}
{\href{https://arxiv.org/abs/gr-qc/9405049}{{arXiv:gr-qc/9405049}}}
\end{barticle}
\endbibitem

%%% 28
\bibitem[\protect\citeauthoryear{Mosani et~al.}{2020}]{PhysRevD.101.044052}
\begin{barticle}
\bauthor{\bsnm{Mosani}, \binits{K.}},
\bauthor{\bsnm{Dey}, \binits{D.}},
\bauthor{\bsnm{Joshi}, \binits{P.S.}}:
\batitle{Strong curvature naked singularities in spherically symmetric perfect fluid collapse}.
\bjtitle{Phys. Rev. D}
\bvolume{101},
\bfpage{044052}
(\byear{2020})
\doiurl{10.1103/PhysRevD.101.044052}
\end{barticle}
\endbibitem

%%% 29
\bibitem[\protect\citeauthoryear{Randall and Sundrum}{1999}]{Randall_1999}
\begin{barticle}
\bauthor{\bsnm{Randall}, \binits{L.}},
\bauthor{\bsnm{Sundrum}, \binits{R.}}:
\batitle{Large mass hierarchy from a small extra dimension}.
\bjtitle{Physical Review Letters}
\bvolume{83}(\bissue{17}),
\bfpage{3370}--\blpage{3373}
(\byear{1999})
\doiurl{10.1103/physrevlett.83.3370}
\end{barticle}
\endbibitem

%%% 30
\bibitem[\protect\citeauthoryear{Horowitz}{2005}]{horowitz2005higher}
\begin{botherref}
\oauthor{\bsnm{Horowitz}, \binits{G.T.}}:
Higher Dimensional Generalizations of the Kerr Black Hole
(2005)
\end{botherref}
\endbibitem

%%% 31
\bibitem[\protect\citeauthoryear{Goswami and Joshi}{2007}]{PhysRevD.76.084026}
\begin{barticle}
\bauthor{\bsnm{Goswami}, \binits{R.}},
\bauthor{\bsnm{Joshi}, \binits{P.S.}}:
\batitle{Spherical gravitational collapse in $n$ dimensions}.
\bjtitle{Phys. Rev. D}
\bvolume{76},
\bfpage{084026}
(\byear{2007})
\doiurl{10.1103/PhysRevD.76.084026}
\end{barticle}
\endbibitem

%%% 32
\bibitem[\protect\citeauthoryear{Giambo and Quintavalle}{2008}]{Giambo:2007ps}
\begin{barticle}
\bauthor{\bsnm{Giambo}, \binits{R.}},
\bauthor{\bsnm{Quintavalle}, \binits{S.}}:
\batitle{{Dimensional dependence of naked singularity formation in spherical gravitational collapse}}.
\bjtitle{Class. Quant. Grav.}
\bvolume{25},
\bfpage{145003}
(\byear{2008})
\doiurl{10.1088/0264-9381/25/14/145003}
{\href{https://arxiv.org/abs/0707.1608}{{arXiv:0707.1608}}}
{[gr-qc]}
\end{barticle}
\endbibitem

%%% 33
\bibitem[\protect\citeauthoryear{Nolan}{2003}]{Nolan_2003}
\begin{barticle}
\bauthor{\bsnm{Nolan}, \binits{B.C.}}:
\batitle{Dynamical extensions for shell-crossing singularities}.
\bjtitle{Classical and Quantum Gravity}
\bvolume{20}(\bissue{4}),
\bfpage{575}--\blpage{585}
(\byear{2003})
\doiurl{10.1088/0264-9381/20/4/302}
\end{barticle}
\endbibitem

%%% 34
\bibitem[\protect\citeauthoryear{Nolan}{1999}]{Nolan:1999tv}
\begin{botherref}
\oauthor{\bsnm{Nolan}, \binits{B.C.}}:
{A Regular $C^0$ singularity is not necessarily weak}
(1999)
{\href{https://arxiv.org/abs/gr-qc/9902020}{{arXiv:gr-qc/9902020}}}
\end{botherref}
\endbibitem

%%% 35
\bibitem[\protect\citeauthoryear{Nolan}{2000}]{Nolan:2000rn}
\begin{barticle}
\bauthor{\bsnm{Nolan}, \binits{B.C.}}:
\batitle{{The Central singularity in spherical collapse}}.
\bjtitle{Phys. Rev. D}
\bvolume{62},
\bfpage{044015}
(\byear{2000})
\doiurl{10.1103/PhysRevD.62.044015}
{\href{https://arxiv.org/abs/gr-qc/0001026}{{arXiv:gr-qc/0001026}}}
\end{barticle}
\endbibitem

%%% 36
\bibitem[\protect\citeauthoryear{Nolan}{1999}]{Nolan:1999tw}
\begin{barticle}
\bauthor{\bsnm{Nolan}, \binits{B.C.}}:
\batitle{{Strengths of singularities in spherical symmetry}}.
\bjtitle{Phys. Rev. D}
\bvolume{60},
\bfpage{024014}
(\byear{1999})
\doiurl{10.1103/PhysRevD.60.024014}
{\href{https://arxiv.org/abs/gr-qc/9902021}{{arXiv:gr-qc/9902021}}}
\end{barticle}
\endbibitem

%%% 37
\bibitem[\protect\citeauthoryear{Clarke}{1993}]{clarke1993analysis}
\begin{bbook}
\bauthor{\bsnm{Clarke}, \binits{C.J.}}:
\bbtitle{The Analysis of Space-time Singularities}
vol. \bseriesno{1}.
\bpublisher{Cambridge University Press},
\blocation{Cambridge}
(\byear{1993})
\end{bbook}
\endbibitem

%%% 38
\bibitem[\protect\citeauthoryear{Bento et~al.}{2002}]{PhysRevD.66.043507}
\begin{barticle}
\bauthor{\bsnm{Bento}, \binits{M.C.}},
\bauthor{\bsnm{Bertolami}, \binits{O.}},
\bauthor{\bsnm{Sen}, \binits{A.A.}}:
\batitle{Generalized chaplygin gas, accelerated expansion, and dark-energy-matter unification}.
\bjtitle{Phys. Rev. D}
\bvolume{66},
\bfpage{043507}
(\byear{2002})
\doiurl{10.1103/PhysRevD.66.043507}
\end{barticle}
\endbibitem

%%% 39
\bibitem[\protect\citeauthoryear{Joshi and Dwivedi}{1992}]{Joshi:1992vr}
\begin{barticle}
\bauthor{\bsnm{Joshi}, \binits{P.S.}},
\bauthor{\bsnm{Dwivedi}, \binits{I.H.}}:
\batitle{{The Structure of Naked Singularity in Self-Similar Gravitational Collapse}}.
\bjtitle{Commun. Math. Phys.}
\bvolume{146},
\bfpage{333}--\blpage{342}
(\byear{1992})
\doiurl{10.1007/BF02102631}
\end{barticle}
\endbibitem

%%% 40
\bibitem[\protect\citeauthoryear{Ori and Piran}{1987}]{PhysRevLett.59.2137}
\begin{barticle}
\bauthor{\bsnm{Ori}, \binits{A.}},
\bauthor{\bsnm{Piran}, \binits{T.}}:
\batitle{Naked singularities in self-similar spherical gravitational collapse}.
\bjtitle{Phys. Rev. Lett.}
\bvolume{59},
\bfpage{2137}--\blpage{2140}
(\byear{1987})
\doiurl{10.1103/PhysRevLett.59.2137}
\end{barticle}
\endbibitem

%%% 41
\bibitem[\protect\citeauthoryear{Lake}{1988}]{PhysRevLett.60.241}
\begin{barticle}
\bauthor{\bsnm{Lake}, \binits{K.}}:
\batitle{Comment on "naked singularities in self-similar spherical gravitational collapse"}.
\bjtitle{Phys. Rev. Lett.}
\bvolume{60},
\bfpage{241}--\blpage{241}
(\byear{1988})
\doiurl{10.1103/PhysRevLett.60.241}
\end{barticle}
\endbibitem

%%% 42
\bibitem[\protect\citeauthoryear{Joshi et~al.}{2011}]{Joshi:2011zm}
\begin{barticle}
\bauthor{\bsnm{Joshi}, \binits{P.S.}},
\bauthor{\bsnm{Malafarina}, \binits{D.}},
\bauthor{\bsnm{Narayan}, \binits{R.}}:
\batitle{{Equilibrium configurations from gravitational collapse}}.
\bjtitle{Class. Quant. Grav.}
\bvolume{28},
\bfpage{235018}
(\byear{2011})
\doiurl{10.1088/0264-9381/28/23/235018}
{\href{https://arxiv.org/abs/1106.5438}{{arXiv:1106.5438}}}
{[gr-qc]}
\end{barticle}
\endbibitem

%%% 43
\bibitem[\protect\citeauthoryear{Mosani and Joshi}{2024}]{mosani2024regular}
\begin{botherref}
\oauthor{\bsnm{Mosani}, \binits{K.}},
\oauthor{\bsnm{Joshi}, \binits{P.S.}}:
Regular black hole from regular initial data
(2024)
\end{botherref}
\endbibitem

%%% 44
\bibitem[\protect\citeauthoryear{Israel}{1966}]{Israel:1966rt}
\begin{barticle}
\bauthor{\bsnm{Israel}, \binits{W.}}:
\batitle{{Singular hypersurfaces and thin shells in general relativity}}.
\bjtitle{Nuovo Cim. B}
\bvolume{44S10},
\bfpage{1}
(\byear{1966})
\doiurl{10.1007/BF02710419} .
\bcomment{[Erratum: Nuovo Cim.B 48, 463 (1967)]}
\end{barticle}
\endbibitem

%%% 45
\bibitem[\protect\citeauthoryear{Wang and Wu}{1999}]{Wang_1999}
\begin{barticle}
\bauthor{\bsnm{Wang}, \binits{A.}},
\bauthor{\bsnm{Wu}, \binits{Y.}}:
\batitle{Letter: Generalized vaidya solutions}.
\bjtitle{General Relativity and Gravitation}
\bvolume{31}(\bissue{1}),
\bfpage{107}--\blpage{114}
(\byear{1999})
\doiurl{10.1023/a:1018819521971}
\end{barticle}
\endbibitem

%%% 46
\bibitem[\protect\citeauthoryear{Mkenyeleye et~al.}{2015}]{PhysRevD.92.024041}
\begin{barticle}
\bauthor{\bsnm{Mkenyeleye}, \binits{M.D.}},
\bauthor{\bsnm{Goswami}, \binits{R.}},
\bauthor{\bsnm{Maharaj}, \binits{S.D.}}:
\batitle{Is cosmic censorship restored in higher dimensions?}
\bjtitle{Phys. Rev. D}
\bvolume{92},
\bfpage{024041}
(\byear{2015})
\doiurl{10.1103/PhysRevD.92.024041}
\end{barticle}
\endbibitem

%%% 47
\bibitem[\protect\citeauthoryear{Mena}{2012}]{Mena:2012zza}
\begin{barticle}
\bauthor{\bsnm{Mena}, \binits{F.C.}}:
\batitle{{Spacetime junctions and the collapse to black holes in higher dimensions}}.
\bjtitle{Adv. Math. Phys.}
\bvolume{2012},
\bfpage{638726}
(\byear{2012})
\doiurl{10.1155/2012/638726}
\end{barticle}
\endbibitem

%%% 48
\bibitem[\protect\citeauthoryear{Sasane}{2022}]{sasane2022mathematical}
\begin{bbook}
\bauthor{\bsnm{Sasane}, \binits{A.}}:
\bbtitle{A Mathematical Introduction to General Relativity}.
\bpublisher{World Scientific},
\blocation{Singapore}
(\byear{2022})
\end{bbook}
\endbibitem

%%% 49
\bibitem[\protect\citeauthoryear{Giambo et~al.}{2003}]{Giambo:2002tp}
\begin{barticle}
\bauthor{\bsnm{Giambo}, \binits{R.}},
\bauthor{\bsnm{Giannoni}, \binits{F.}},
\bauthor{\bsnm{Magli}, \binits{G.}},
\bauthor{\bsnm{Piccione}, \binits{P.}}:
\batitle{{New mathematical framework for spherical gravitational collapse}}.
\bjtitle{Class. Quant. Grav.}
\bvolume{20},
\bfpage{75}
(\byear{2003})
\doiurl{10.1088/0264-9381/20/6/102}
{\href{https://arxiv.org/abs/gr-qc/0212082}{{arXiv:gr-qc/0212082}}}
\end{barticle}
\endbibitem

%%% 50
\bibitem[\protect\citeauthoryear{Faraoni et~al.}{2021}]{Faraoni:2020mdf}
\begin{barticle}
\bauthor{\bsnm{Faraoni}, \binits{V.}},
\bauthor{\bsnm{Giusti}, \binits{A.}},
\bauthor{\bsnm{Bean}, \binits{T.F.}}:
\batitle{{Asymptotic flatness and Hawking quasilocal mass}}.
\bjtitle{Phys. Rev. D}
\bvolume{103}(\bissue{4}),
\bfpage{044026}
(\byear{2021})
\doiurl{10.1103/PhysRevD.103.044026}
{\href{https://arxiv.org/abs/2010.00069}{{arXiv:2010.00069}}}
{[gr-qc]}
\end{barticle}
\endbibitem

%%% 51
\bibitem[\protect\citeauthoryear{Hawking}{1968}]{Hawking:1968qt}
\begin{barticle}
\bauthor{\bsnm{Hawking}, \binits{S.}}:
\batitle{{Gravitational radiation in an expanding universe}}.
\bjtitle{J. Math. Phys.}
\bvolume{9},
\bfpage{598}--\blpage{604}
(\byear{1968})
\doiurl{10.1063/1.1664615}
\end{barticle}
\endbibitem

%%% 52
\bibitem[\protect\citeauthoryear{Arnowitt et~al.}{2008}]{Arnowitt:1962hi}
\begin{barticle}
\bauthor{\bsnm{Arnowitt}, \binits{R.L.}},
\bauthor{\bsnm{Deser}, \binits{S.}},
\bauthor{\bsnm{Misner}, \binits{C.W.}}:
\batitle{{The Dynamics of general relativity}}.
\bjtitle{Gen. Rel. Grav.}
\bvolume{40},
\bfpage{1997}--\blpage{2027}
(\byear{2008})
\doiurl{10.1007/s10714-008-0661-1}
{\href{https://arxiv.org/abs/gr-qc/0405109}{{arXiv:gr-qc/0405109}}}
\end{barticle}
\endbibitem

%%% 53
\bibitem[\protect\citeauthoryear{Szabados}{2009}]{Szabados:2009eka}
\begin{barticle}
\bauthor{\bsnm{Szabados}, \binits{L.B.}}:
\batitle{{Quasi-Local Energy-Momentum and Angular Momentum in General Relativity}}.
\bjtitle{Living Rev. Rel.}
\bvolume{12},
\bfpage{4}
(\byear{2009})
\doiurl{10.12942/lrr-2009-4}
\end{barticle}
\endbibitem

%%% 54
\bibitem[\protect\citeauthoryear{Hayward}{1994}]{PhysRevD.49.831}
\begin{barticle}
\bauthor{\bsnm{Hayward}, \binits{S.A.}}:
\batitle{Quasilocal gravitational energy}.
\bjtitle{Phys. Rev. D}
\bvolume{49},
\bfpage{831}--\blpage{839}
(\byear{1994})
\doiurl{10.1103/PhysRevD.49.831}
\end{barticle}
\endbibitem

%%% 55
\bibitem[\protect\citeauthoryear{Schoen and Yau}{1979}]{Schoen:1979zz}
\begin{barticle}
\bauthor{\bsnm{Schoen}, \binits{R.}},
\bauthor{\bsnm{Yau}, \binits{S.-T.}}:
\batitle{{Positivity of the Total Mass of a General Space-Time}}.
\bjtitle{Phys. Rev. Lett.}
\bvolume{43},
\bfpage{1457}--\blpage{1459}
(\byear{1979})
\doiurl{10.1103/PhysRevLett.43.1457}
\end{barticle}
\endbibitem

%%% 56
\bibitem[\protect\citeauthoryear{Witten}{1981}]{Witten:1981mf}
\begin{barticle}
\bauthor{\bsnm{Witten}, \binits{E.}}:
\batitle{{A Simple Proof of the Positive Energy Theorem}}.
\bjtitle{Commun. Math. Phys.}
\bvolume{80},
\bfpage{381}
(\byear{1981})
\doiurl{10.1007/BF01208277}
\end{barticle}
\endbibitem

%%% 57
\bibitem[\protect\citeauthoryear{Szekeres and Lun}{1999}]{Szekeres:1995gy}
\begin{barticle}
\bauthor{\bsnm{Szekeres}, \binits{P.}},
\bauthor{\bsnm{Lun}, \binits{A.}}:
\batitle{{What is a shell crossing singularity?}}
\bjtitle{J. Austral. Math. Soc. B}
\bvolume{41},
\bfpage{167}--\blpage{179}
(\byear{1999})
\doiurl{10.1017/S0334270000011140}
\end{barticle}
\endbibitem

%%% 58
\bibitem[\protect\citeauthoryear{Scott and Szekeres}{1994}]{Scott:1994sn}
\begin{barticle}
\bauthor{\bsnm{Scott}, \binits{S.M.}},
\bauthor{\bsnm{Szekeres}, \binits{P.}}:
\batitle{{The Abstract boundary: A new approach to singularities of manifolds}}.
\bjtitle{J. Geom. Phys.}
\bvolume{13},
\bfpage{223}--\blpage{253}
(\byear{1994})
\doiurl{10.1016/0393-0440(94)90032-9}
{\href{https://arxiv.org/abs/gr-qc/9405063}{{arXiv:gr-qc/9405063}}}
\end{barticle}
\endbibitem

%%% 59
\bibitem[\protect\citeauthoryear{Joshi and Saraykar}{2013}]{Joshi:2012ak}
\begin{barticle}
\bauthor{\bsnm{Joshi}, \binits{P.S.}},
\bauthor{\bsnm{Saraykar}, \binits{R.V.}}:
\batitle{{Shell-crossings in Gravitational Collapse}}.
\bjtitle{Int. J. Mod. Phys. D}
\bvolume{22},
\bfpage{1350027}
(\byear{2013})
\doiurl{10.1142/S0218271813500272}
{\href{https://arxiv.org/abs/1205.3263}{{arXiv:1205.3263}}}
{[gr-qc]}
\end{barticle}
\endbibitem

%%% 60
\bibitem[\protect\citeauthoryear{Pretel and da Silva}{2020}]{10.1093/mnras/staa1493}
\begin{barticle}
\bauthor{\bsnm{Pretel}, \binits{J.M.Z.}},
\bauthor{\bsnm{da Silva}, \binits{M.F.A.}}:
\batitle{{Stability and gravitational collapse of neutron stars with realistic equations of state}}.
\bjtitle{Monthly Notices of the Royal Astronomical Society}
\bvolume{495}(\bissue{4}),
\bfpage{5027}--\blpage{5039}
(\byear{2020})
\doiurl{10.1093/mnras/staa1493}
{\href{https://arxiv.org/abs/https://academic.oup.com/mnras/article-pdf/495/4/5027/33386239/staa1493.pdf}{{https://academic.oup.com/mnras/article-pdf/495/4/5027/33386239/staa1493.pdf}}}
\end{barticle}
\endbibitem

%%% 61
\bibitem[\protect\citeauthoryear{Ahmad et~al.}{2018}]{Ahmad_2018}
\begin{barticle}
\bauthor{\bsnm{Ahmad}, \binits{Z.}},
\bauthor{\bsnm{Shah}, \binits{H.}},
\bauthor{\bsnm{Khan}, \binits{S.}}:
\batitle{Spherical gravitational collapse in f(r) gravity with linear equation of state}.
\bjtitle{Communications in Theoretical Physics}
\bvolume{70}(\bissue{2}),
\bfpage{185}
(\byear{2018})
\doiurl{10.1088/0253-6102/70/2/185}
\end{barticle}
\endbibitem

%%% 62
\bibitem[\protect\citeauthoryear{Goswami and Joshi}{2004}]{Goswami:2004kq}
\begin{barticle}
\bauthor{\bsnm{Goswami}, \binits{R.}},
\bauthor{\bsnm{Joshi}, \binits{P.S.}}:
\batitle{{Gravitational collapse of an isentropic perfect fluid with a linear equation of state}}.
\bjtitle{Class. Quant. Grav.}
\bvolume{21},
\bfpage{3645}--\blpage{3654}
(\byear{2004})
\doiurl{10.1088/0264-9381/21/15/002}
{\href{https://arxiv.org/abs/gr-qc/0406052}{{arXiv:gr-qc/0406052}}}
\end{barticle}
\endbibitem

%%% 63
\bibitem[\protect\citeauthoryear{Chru{\'s}ciel}{1991}]{chruściel1991uniqueness}
\begin{botherref}
\oauthor{\bsnm{Chru{\'s}ciel}, \binits{P.T.}}:
On Uniqueness in the Large of Solutions of Einstein's Equations ("Strong Cosmic Censorship").
Proceedings of the Centre for Mathematics and Its Applications, Australian National University,
Australia
(1991)
\end{botherref}
\endbibitem

%%% 64
\bibitem[\protect\citeauthoryear{Altas and Tekin}{2022}]{Altas_2022}
\begin{barticle}
\bauthor{\bsnm{Altas}, \binits{E.}},
\bauthor{\bsnm{Tekin}, \binits{B.}}:
\batitle{Basics of apparent horizons in black hole physics}.
\bjtitle{Journal of Physics: Conference Series}
\bvolume{2191}(\bissue{1}),
\bfpage{012002}
(\byear{2022})
\doiurl{10.1088/1742-6596/2191/1/012002}
\end{barticle}
\endbibitem

%%% 65
\bibitem[\protect\citeauthoryear{Dotti}{2024}]{Dotti:2023elh}
\begin{barticle}
\bauthor{\bsnm{Dotti}, \binits{G.}}:
\batitle{{Black hole regions containing no trapped surfaces}}.
\bjtitle{Class. Quant. Grav.}
\bvolume{41}(\bissue{1}),
\bfpage{015015}
(\byear{2024})
\doiurl{10.1088/1361-6382/ad0fb9}
{\href{https://arxiv.org/abs/2308.13950}{{arXiv:2308.13950}}}
{[gr-qc]}
\end{barticle}
\endbibitem

%%% 66
\bibitem[\protect\citeauthoryear{Geroch et~al.}{1972}]{21afc2b9-b08c-31a5-a52e-6412f7dd10e6}
\begin{barticle}
\bauthor{\bsnm{Geroch}, \binits{R.}},
\bauthor{\bsnm{Kronheimer}, \binits{E.H.}},
\bauthor{\bsnm{Penrose}, \binits{R.}}:
\batitle{Ideal points in space-time}.
\bjtitle{Proceedings of the Royal Society of London. Series A, Mathematical and Physical Sciences}
\bvolume{327}(\bissue{1571}),
\bfpage{545}--\blpage{567}
(\byear{1972}).
Accessed 2023-12-06
\end{barticle}
\endbibitem

%%% 67
\bibitem[\protect\citeauthoryear{Ellis and Schmidt}{1977}]{Ellis:1977pj}
\begin{barticle}
\bauthor{\bsnm{Ellis}, \binits{G.F.R.}},
\bauthor{\bsnm{Schmidt}, \binits{B.G.}}:
\batitle{{Singular space-times}}.
\bjtitle{Gen. Rel. Grav.}
\bvolume{8},
\bfpage{915}--\blpage{953}
(\byear{1977})
\doiurl{10.1007/BF00759240}
\end{barticle}
\endbibitem

%%% 68
\bibitem[\protect\citeauthoryear{Hawking and Ellis}{2023}]{hawking2023large}
\begin{bbook}
\bauthor{\bsnm{Hawking}, \binits{S.W.}},
\bauthor{\bsnm{Ellis}, \binits{G.F.}}:
\bbtitle{The Large Scale Structure of Space-time}.
\bpublisher{Cambridge university press},
\blocation{Cambridge}
(\byear{2023})
\end{bbook}
\endbibitem

%%% 69
\bibitem[\protect\citeauthoryear{Maeda and Martinez}{2020}]{Maeda:2018hqu}
\begin{barticle}
\bauthor{\bsnm{Maeda}, \binits{H.}},
\bauthor{\bsnm{Martinez}, \binits{C.}}:
\batitle{{Energy conditions in arbitrary dimensions}}.
\bjtitle{PTEP}
\bvolume{2020}(\bissue{4}),
\bfpage{043}--\blpage{02}
(\byear{2020})
\doiurl{10.1093/ptep/ptaa009}
{\href{https://arxiv.org/abs/1810.02487}{{arXiv:1810.02487}}}
{[gr-qc]}
\end{barticle}
\endbibitem

%%% 70
\bibitem[\protect\citeauthoryear{Abraham and Marsden}{2008}]{abraham2008foundations}
\begin{bbook}
\bauthor{\bsnm{Abraham}, \binits{R.}},
\bauthor{\bsnm{Marsden}, \binits{J.E.}}:
\bbtitle{Foundations of Mechanics}
vol. \bseriesno{364}.
\bpublisher{American Mathematical Soc.},
\blocation{USA}
(\byear{2008})
\end{bbook}
\endbibitem

%%% 71
\bibitem[\protect\citeauthoryear{Goswami and Joshi}{2004}]{PhysRevD.69.104002}
\begin{barticle}
\bauthor{\bsnm{Goswami}, \binits{R.}},
\bauthor{\bsnm{Joshi}, \binits{P.S.}}:
\batitle{Cosmic censorship in higher dimensions}.
\bjtitle{Phys. Rev. D}
\bvolume{69},
\bfpage{104002}
(\byear{2004})
\doiurl{10.1103/PhysRevD.69.104002}
\end{barticle}
\endbibitem

%%% 72
\bibitem[\protect\citeauthoryear{Mahajan et~al.}{2005}]{PhysRevD.72.024006}
\begin{barticle}
\bauthor{\bsnm{Mahajan}, \binits{A.}},
\bauthor{\bsnm{Goswami}, \binits{R.}},
\bauthor{\bsnm{Joshi}, \binits{P.S.}}:
\batitle{Cosmic censorship in higher dimensions. ii.}
\bjtitle{Phys. Rev. D}
\bvolume{72},
\bfpage{024006}
(\byear{2005})
\doiurl{10.1103/PhysRevD.72.024006}
\end{barticle}
\endbibitem

%%% 73
\bibitem[\protect\citeauthoryear{Banerjee et~al.}{2003}]{Banerjee:2002sy}
\begin{barticle}
\bauthor{\bsnm{Banerjee}, \binits{A.}},
\bauthor{\bsnm{Debnath}, \binits{U.}},
\bauthor{\bsnm{Chakraborty}, \binits{S.}}:
\batitle{{Naked singularities in higher dimensional gravitational collapse}}.
\bjtitle{Int. J. Mod. Phys. D}
\bvolume{12},
\bfpage{1255}--\blpage{1264}
(\byear{2003})
\doiurl{10.1142/S021827180300375X}
{\href{https://arxiv.org/abs/gr-qc/0211099}{{arXiv:gr-qc/0211099}}}
\end{barticle}
\endbibitem

%%% 74
\bibitem[\protect\citeauthoryear{Chakraborty and Sengupta}{2018}]{Chakraborty:2017uku}
\begin{barticle}
\bauthor{\bsnm{Chakraborty}, \binits{S.}},
\bauthor{\bsnm{Sengupta}, \binits{S.}}:
\batitle{{Packing extra mass in compact stellar structures: An interplay between Kalb-Ramond field and extra dimensions}}.
\bjtitle{JCAP}
\bvolume{05},
\bfpage{032}
(\byear{2018})
\doiurl{10.1088/1475-7516/2018/05/032}
{\href{https://arxiv.org/abs/1708.08315}{{arXiv:1708.08315}}}
{[gr-qc]}
\end{barticle}
\endbibitem

%%% 75
\bibitem[\protect\citeauthoryear{Alday and Perlmutter}{2019}]{Alday:2019qrf}
\begin{barticle}
\bauthor{\bsnm{Alday}, \binits{L.F.}},
\bauthor{\bsnm{Perlmutter}, \binits{E.}}:
\batitle{{Growing Extra Dimensions in AdS/CFT}}.
\bjtitle{JHEP}
\bvolume{08},
\bfpage{084}
(\byear{2019})
\doiurl{10.1007/JHEP08(2019)084}
{\href{https://arxiv.org/abs/1906.01477}{{arXiv:1906.01477}}}
{[hep-th]}
\end{barticle}
\endbibitem

%%% 76
\bibitem[\protect\citeauthoryear{Lu et~al.}{2009}]{Lu:2008jk}
\begin{barticle}
\bauthor{\bsnm{Lu}, \binits{H.}},
\bauthor{\bsnm{Mei}, \binits{J.}},
\bauthor{\bsnm{Pope}, \binits{C.N.}}:
\batitle{{Kerr/CFT Correspondence in Diverse Dimensions}}.
\bjtitle{JHEP}
\bvolume{04},
\bfpage{054}
(\byear{2009})
\doiurl{10.1088/1126-6708/2009/04/054}
{\href{https://arxiv.org/abs/0811.2225}{{arXiv:0811.2225}}}
{[hep-th]}
\end{barticle}
\endbibitem

%%% 77
\bibitem[\protect\citeauthoryear{Frassino et~al.}{2023}]{Frassino:2022zaz}
\begin{barticle}
\bauthor{\bsnm{Frassino}, \binits{A.M.}},
\bauthor{\bsnm{Pedraza}, \binits{J.F.}},
\bauthor{\bsnm{Svesko}, \binits{A.}},
\bauthor{\bsnm{Visser}, \binits{M.R.}}:
\batitle{{Higher-Dimensional Origin of Extended Black Hole Thermodynamics}}.
\bjtitle{Phys. Rev. Lett.}
\bvolume{130}(\bissue{16}),
\bfpage{161501}
(\byear{2023})
\doiurl{10.1103/PhysRevLett.130.161501}
{\href{https://arxiv.org/abs/2212.14055}{{arXiv:2212.14055}}}
{[hep-th]}
\end{barticle}
\endbibitem

%%% 78
\bibitem[\protect\citeauthoryear{Pourhassan et~al.}{2017}]{Pourhassan:2017kmm}
\begin{barticle}
\bauthor{\bsnm{Pourhassan}, \binits{B.}},
\bauthor{\bsnm{Kokabi}, \binits{K.}},
\bauthor{\bsnm{Rangyan}, \binits{S.}}:
\batitle{{Thermodynamics of higher dimensional black holes with higher order thermal fluctuations}}.
\bjtitle{Gen. Rel. Grav.}
\bvolume{49}(\bissue{12}),
\bfpage{144}
(\byear{2017})
\doiurl{10.1007/s10714-017-2315-7}
{\href{https://arxiv.org/abs/1710.06299}{{arXiv:1710.06299}}}
{[gr-qc]}
\end{barticle}
\endbibitem

%%% 79
\bibitem[\protect\citeauthoryear{Gomez-Fayren et~al.}{2023}]{Gomez-Fayren:2023wxk}
\begin{barticle}
\bauthor{\bsnm{Gomez-Fayren}, \binits{C.}},
\bauthor{\bsnm{Meessen}, \binits{P.}},
\bauthor{\bsnm{Ortin}, \binits{T.}},
\bauthor{\bsnm{Zatti}, \binits{M.}}:
\batitle{{Wald entropy in Kaluza-Klein black holes}}.
\bjtitle{JHEP}
\bvolume{08},
\bfpage{039}
(\byear{2023})
\doiurl{10.1007/JHEP08(2023)039}
{\href{https://arxiv.org/abs/2305.01742}{{arXiv:2305.01742}}}
{[hep-th]}
\end{barticle}
\endbibitem

%%% 80
\bibitem[\protect\citeauthoryear{Hendi et~al.}{2021}]{Hendi:2021yii}
\begin{barticle}
\bauthor{\bsnm{Hendi}, \binits{S.H.}},
\bauthor{\bsnm{Hajkhalili}, \binits{S.}},
\bauthor{\bsnm{Jamil}, \binits{M.}},
\bauthor{\bsnm{Momennia}, \binits{M.}}:
\batitle{{Stability and phase transition of rotating Kaluza\textendash{}Klein black holes}}.
\bjtitle{Eur. Phys. J. C}
\bvolume{81}(\bissue{12}),
\bfpage{1112}
(\byear{2021})
\doiurl{10.1140/epjc/s10052-021-09836-9}
{\href{https://arxiv.org/abs/2111.10117}{{arXiv:2111.10117}}}
{[gr-qc]}
\end{barticle}
\endbibitem

%%% 81
\bibitem[\protect\citeauthoryear{Deo et~al.}{2023}]{Deo:2023vvb}
\begin{barticle}
\bauthor{\bsnm{Deo}, \binits{I.}},
\bauthor{\bsnm{Dhivakar}, \binits{P.}},
\bauthor{\bsnm{Kundu}, \binits{N.}}:
\batitle{{Entropy-current for dynamical black holes in Chern-Simons theories of gravity}}.
\bjtitle{JHEP}
\bvolume{11},
\bfpage{114}
(\byear{2023})
\doiurl{10.1007/JHEP11(2023)114}
{\href{https://arxiv.org/abs/2306.12491}{{arXiv:2306.12491}}}
{[hep-th]}
\end{barticle}
\endbibitem

%%% 82
\bibitem[\protect\citeauthoryear{Goswami et~al.}{2006}]{PhysRevLett.96.031302}
\begin{barticle}
\bauthor{\bsnm{Goswami}, \binits{R.}},
\bauthor{\bsnm{Joshi}, \binits{P.S.}},
\bauthor{\bsnm{Singh}, \binits{P.}}:
\batitle{Quantum evaporation of a naked singularity}.
\bjtitle{Phys. Rev. Lett.}
\bvolume{96},
\bfpage{031302}
(\byear{2006})
\doiurl{10.1103/PhysRevLett.96.031302}
\end{barticle}
\endbibitem

%%% 83
\bibitem[\protect\citeauthoryear{Emparan and Reall}{2008}]{Emparan:2008eg}
\begin{barticle}
\bauthor{\bsnm{Emparan}, \binits{R.}},
\bauthor{\bsnm{Reall}, \binits{H.S.}}:
\batitle{{Black Holes in Higher Dimensions}}.
\bjtitle{Living Rev. Rel.}
\bvolume{11},
\bfpage{6}
(\byear{2008})
\doiurl{10.12942/lrr-2008-6}
{\href{https://arxiv.org/abs/0801.3471}{{arXiv:0801.3471}}}
{[hep-th]}
\end{barticle}
\endbibitem

%%% 84
\bibitem[\protect\citeauthoryear{Maldacena}{1998}]{Maldacena:1997re}
\begin{barticle}
\bauthor{\bsnm{Maldacena}, \binits{J.M.}}:
\batitle{{The Large N limit of superconformal field theories and supergravity}}.
\bjtitle{Adv. Theor. Math. Phys.}
\bvolume{2},
\bfpage{231}--\blpage{252}
(\byear{1998})
\doiurl{10.4310/ATMP.1998.v2.n2.a1}
{\href{https://arxiv.org/abs/hep-th/9711200}{{arXiv:hep-th/9711200}}}
\end{barticle}
\endbibitem

%%% 85
\bibitem[\protect\citeauthoryear{Gubser}{2000}]{Gubser:2000nd}
\begin{barticle}
\bauthor{\bsnm{Gubser}, \binits{S.S.}}:
\batitle{{Curvature singularities: The Good, the bad, and the naked}}.
\bjtitle{Adv. Theor. Math. Phys.}
\bvolume{4},
\bfpage{679}--\blpage{745}
(\byear{2000})
\doiurl{10.4310/ATMP.2000.v4.n3.a6}
{\href{https://arxiv.org/abs/hep-th/0002160}{{arXiv:hep-th/0002160}}}
\end{barticle}
\endbibitem

\end{thebibliography}

\end{document}